\documentclass[aps,pra,onecolumn,nofootinbib,floatfix,longbibliography,groupedaddress, 10pt]{revtex4-2}
\usepackage{algorithm2e,amsfonts,amsthm,array,bbm,bm,color,fancyvrb,graphicx,hyperref,lmodern,mathtools,nicematrix,pgfplots,placeins,tikz,xparse,xfrac, amssymb}
\usepackage[margin=2cm]{geometry}
\usepackage[normalem]{ulem}
\usepackage{float}
\usepackage{subfig}
\usetikzlibrary{quantikz,calc,arrows,shapes,shapes.geometric,positioning,calc,decorations,decorations.pathmorphing,decorations.text,decorations.markings,fadings,decorations.pathreplacing,matrix,tikzmark,fit}
\tikzset{%
  >=latex, 
  inner sep=0pt,%
  outer sep=2pt,%
  mark coordinate/.style={inner sep=1pt,outer sep=1pt,minimum size=3pt,
    fill=black}
}

\pgfplotsset{compat=1.18}
\hypersetup{colorlinks=true,linkcolor=blue,citecolor=blue,urlcolor=blue}
\NiceMatrixOptions{cell-space-limits = 2mm}

\DeclarePairedDelimiter\ceil{\lceil}{\rceil}
\DeclarePairedDelimiter\floor{\lfloor}{\rfloor}

\newcommand{\comment}[1]{}

\newcommand{\mycomment}[1]{{}}

\newcommand{\abs}[1]{{\left\lvert{#1}\right\rvert}}
\newcommand{\norm}[1]{{\left\lVert{#1}\right\rVert}}

\DeclareDocumentCommand{\cyc}{ O{d} }{
  {\mathbb{Z}_d}
}
\DeclareDocumentCommand{\pnorm}{ m O{p} }{
  \norm{{#1}}_{#2}
}

\DeclareMathOperator{\tr}{tr}

\newtheorem{thm}{Theorem}
\newtheorem{lem}{Lemma}

\newtheorem{cor}{Corollary}
\newtheorem{exa}{Example}
\newtheorem{res}{Result}
\newtheorem{app}{Application}
\newtheorem{rem}{Remark}
\theoremstyle{definition}

\DeclareMathOperator{\erfc}{erfc}
\DeclareMathOperator*{\prob}{Pr}

\begin{document}

\date{\today}
\title{Fundamental solutions of heat equation on unitary groups establish an improved relation between \texorpdfstring{$\epsilon$}{epsilon}-nets and approximate unitary \texorpdfstring{$t$}{t}-designs}

\author{Oskar S{\l}owik }
\email{oslowik@cft.edu.pl}
\affiliation{Center for Theoretical Physics, Polish Academy of Sciences,\\ Al. Lotnik\'ow 32/46, 02-668 Warszawa, Poland}
\author{Oliver Reardon-Smith}
\affiliation{Center for Theoretical Physics, Polish Academy of Sciences,\\ Al. Lotnik\'ow 32/46, 02-668 Warszawa, Poland}
\author{Adam Sawicki}
\affiliation{Center for Theoretical Physics, Polish Academy of Sciences,\\ Al. Lotnik\'ow 32/46, 02-668 Warszawa, Poland}
\affiliation{Guangdong Technion - Israel Institute of Technology, 241 Daxue Road,
Jinping District, Shantou, Guangdong Province, China}

\begin{abstract}
 The concepts of $\epsilon$-nets and unitary ($\delta$-approximate) $t$-designs are important and ubiquitous across quantum computation and information. Both notions are closely related and the quantitative relations between $t$, $\delta$ and $\epsilon$ find applications in areas such as (non-constructive) inverse-free Solovay-Kitaev like theorems and random quantum circuits. In recent work, quantitative relations have revealed the close connection between the two constructions, with $\epsilon$-nets functioning as unitary $\delta$-approximate $t$-designs and vice-versa, for appropriate choice of parameters. In this work we improve these results, significantly increasing the bound on the $\delta$ required for a $\delta$-approximate $t$-design to form an $\epsilon$-net from $\delta \simeq \left(\epsilon^{3/2}/d\right)^{d^2}$ to $\delta \simeq \left(\epsilon/d^{1/2}\right)^{d^2}$.
 We achieve this by constructing polynomial approximations to the Dirac delta using heat kernels on the projective unitary group $\mathrm{PU}(d) \cong\mathbf{U}(d)$, whose properties we studied and which may be applicable more broadly. We also outline the possible applications of our results in quantum circuit overheads, quantum complexity and black hole physics.

\end{abstract}

\maketitle

\section{Introduction}

Unitary $t$-designs are a fundamental construction, finding widespread applications across quantum information and computation. They have been employed in areas such as randomised benchmarking \cite{Wallman_2014, Epstein_2014}, process tomography \cite{Scott_2008}, quantum information protocols \cite{PhysRevA.80.012304, Abeyesinghe_2009}, unitary codes \cite{Roy_2009}, derandomisation of probabilistic constructions \cite{Gross_2014}, decoupling \cite{Szehr_2013}, entanglement detection \cite{Bae_2019}, quantum state discrimination \cite{sen2005randommeasurementbasesquantum}, shadow estimation \cite{Helsen_2023}, efficient quantum measurements \cite{physRevLett.124.090503} and estimation of the properties of quantum systems \cite{Huang_2020}. Moreover, their link to pseudo-random quantum circuits \cite{Brand_o_2016} makes them applicable in areas such as the equilibration
of quantum systems \cite{Masanes_2013, Brand_o_2016}, quantum metrology with random bosonic states \cite{PhysRevX.6.041044}, quantum complexity and information scrambling in black holes \cite{Roberts_2017,Nakata_2017,Brand_o_2021,PhysRevX.14.041068}. They have also been applied to the study of quantum speed-ups \cite{Harrow_2017,Boixo_2018,Arute2019}, due to their anti-concentration property \cite{Hangleiter2018anticoncentration, Yoganathan_2019}.

Epsilon-nets are of similar importance, finding broad application and, in particular, serving as the natural language for quantum compilation. Solovay-Kitaev like (SKL) theorems \cite{kuperberg2023breakingcubicbarriersolovaykitaev,Slowik_qco_2024} provide joint bounds on the complexity of quantum operations $\mathbf{U}$, for a given error $\epsilon$ and gateset $\mathcal{S}$. In other words, they bound how many operations are required for circuits of gates from a given gateset to form an $\epsilon$-net. Moreover, constructive SKL theorems say how to find the approximating circuits, which makes them the cornerstone of quantum compilation. The original SK theorem, which is constructive, bounds the length of the sequence of gates as $\ell=\mathcal{O}({\mathrm{log}^c\left(\frac{1}{\epsilon}\right)})$, where $c\approx 3.97$. In fact, it is well-known that any $c > 3$ works and recently a constructive SKL theorem with $c \approx 1.44 $ was provided in Ref.~\cite{kuperberg2023breakingcubicbarriersolovaykitaev}, which is significantly closer to the optimal value $c=1$.

The parameter $\delta$ of the (unitary) $\delta$-approximate $t$-design generated by $\mathcal{S}$ can be studied on finite scales, say given by the highest considered degree $t$, denoted $\delta(\nu_{\mathcal{S}},t)$. Such a finite-scale approach was explored in Ref.~\cite{9614165,varjú2015randomwalkscompactgroups}. For fixed $\epsilon$ and $\mathcal{S}$, the knowledge of $\delta(\nu_{\mathcal{S}},t)$ at a suitably chosen scale $t(\epsilon)$, is sufficient to bound $\ell$ via a non-constructive SKL theorem $\ell =\mathcal{O}( \frac{1}{\mathrm{log}\left(1/\delta\left(\nu_{\mathcal{S},t(\epsilon)}\right)\right)}  \mathrm{log}\left(\frac{1}{\epsilon}\right))$ with explicit form (see e.g. Ref.~\cite{9614165}). Such SKL theorems can be used to bound the efficiency of various gate sets $\mathcal{S}$, e.g. their ($T$-)Quantum Circuit Overhead \cite{Slowik_qco_2024}. In particular, if the supremum of  $\delta(\nu_{\mathcal{S}},t)$ over all $t$ is smaller than $1$, then we obtain an asymptotically optimal scaling $\ell = \Theta(\mathrm{log}(\frac{1}{\epsilon}))$ \cite{Harrow_2002, S_owik_2023}. However, the analysis of such a supremum is a hard problem and is computationally intractable. Hence, the SKL theorems based on a finite-scale $\delta(\nu_{\mathcal{S}},t(\epsilon))$ are of significant practical interest. The tightness of such theorems depends on the tightness of the $t(\epsilon)$ scaling, which can be understood as the $t$ sufficient for a ($\delta$-approximate) $t$-design to form an $\epsilon$-net.

In light of the importance of both $\epsilon$-nets and $t$-designs, it is interesting that there is a strong link between the two constructions. Indeed a (possibly approximate) $t$-design of sufficiently large $t$ forms an $\epsilon$-net, while an $\epsilon$-net of sufficiently small $\epsilon$ forms an approximate $t$-design. To our knowledge, the first systematic study of the quantitative relations between them was surprisingly recent, in Ref.~\cite{9614165}, where the authors show that an $\epsilon$-net is formed by $\delta$-approximate $t$-designs on the space of unitary channels $\mathbf{U}(d)$ for
$t \simeq \frac{d^{5/2}}{\epsilon}$ and $\delta \simeq \left(\frac{\epsilon^{3/2}}{d}\right)^{d^2}$,
where $\simeq$ can be understood as ``ignoring logarithmic factors'' and ``infinitesimal corrections to the exponents''. The authors of Ref.~\cite{9614165} were able to prove that $t$ has to grow at least as $1/\epsilon$ (for fixed $d$) and as $d^2$ (for fixed $\epsilon$). Thus they were able to show that this scaling of $t$ with $\epsilon$ is essentially optimal, while the scaling of $t$ with $d$ is (at worst) not very far from optimal, with a ``gap'' of $\sqrt{d}$ between the known lower and upper bounds. They conjectured that a scaling of $t \simeq d^{2}$ was possible but were not able to prove this.

In this work we build on these results, obtaining (up to logarithmic factors) the same scaling of $t$ as the authors of Ref.~\cite{9614165}, but dramatically improving the scaling of $\delta$ with $\epsilon$ and $d$ in the $\delta$-approximate case. We are able to show that a $\delta$-approximate $t$-design forms an $\epsilon$-net if delta obeys an inequality which scales like $\delta\simeq\left(\frac{\epsilon}{ d^{\frac{1}{2}}}\right)^{d^2}$.

Our method involves the construction of a polynomial approximation to a Dirac delta on the space of quantum unitary channels $\mathrm{PU}(d)$. Our construction of the approximate Dirac delta is a natural one, based on the properties of the heat kernel on $\mathrm{SU}(d)$. As running the evolution of the heat equation ``forwards'' leads to ``heat'' spreading out over time, naturally running it backward and considering very small times leads to a sharp delta-like peak at times close to $0$. As has been known since the work of Fourier himself, the heat equation has a close link to Fourier analysis on the appropriate space. Indeed, our key bounds are based on the results from Ref.~\cite{urakawa-1974}, which may be viewed as a generalization of the well-known Poisson formula~\cite{Zygmund-trig-textbook} to compact semi-simple simply-connected Lie groups. The heat kernel is at the heart of many important methods across mathematical physics and beyond. As a tool to study the eigenvalues and eigenfunctions of the Laplacian, it has been used as far back as Kac's seminal 1966 paper ``Can One Hear the Shape of a Drum''~\cite{kak-1962} and has been of prime importance in the study of the Laplacian on Riemannian manifolds throughout the subject's history~\cite{Eigenvalues-reimannian}. It has been applied extensively in the context of quantum field theory~\cite{VASSILEVICH2003279}, mathematical finance~\cite{Avramidi2015} and quantum gravity~\cite{Avramidi2000}
and used to prove a diverse range of important theorems, including the Atiyah–Singer index theorem~\cite{Atiyah-1973,Berline2004} and the Poincar{\'e} conjecture~\cite{Perelman-2002,Poincare-Conjecture}. For a thorough review, we invite the reader to the textbooks~\cite{Avramidi2015,heat-kernel-manifolds-textbook}.

Our core results - the bounds on $t$ and $\delta$ are the subjects of Theorem \ref{th:theorem1} (for $t$-designs) and Theorem \ref{th:theorem2} (for $\delta$-approximate $t$-designs). We also provide a technical result about the properties of our approximate Dirac delta (Theorem \ref{th:theorem3}), which may be useful for other applications.

\textbf{Outline of the proof} - the proof of the main theorem (Theorem \ref{th:theorem2}) can be divided into five steps:

\begin{enumerate}
    \item We ``trim'' the full heat kernel on $\mathrm{PU}(d)$ to obtain an approximation of it by a balanced polynomial of order $t$, and prove an error bound for this approximation.
    \item We prove that the heat kernel on $\mathrm{PU}(d)$ is an approximation to the Dirac delta. In particular, its integral vanishes outside any $\epsilon$-ball as $\sigma\to 0$ at a rate we can bound.
    \item By combining the above two bounds, we obtain a bound for the integral of the absolute value of the trimmed heat kernel outside an $\epsilon$-ball, thereby showing the trimmed heat kernel is also an approximation to the Dirac delta.
    \item We derive the bounds on the $L^2$-norm of the heat kernel on $\mathrm{PU}(d)$.
    \item We combine the bounds to obtain a bound for the $t$ and $\delta$ sufficient for a projective unitary $\delta$-approximate $t$-design to be an $\epsilon$-net. Essentially, this argument follows from applying the $t$-design property to the order $t$ balanced polynomial we obtained in step $1$.
\end{enumerate}


\textbf{Structure of the paper} - the paper is organised as follows:

\begin{itemize}
\item In Section~\ref{sec:intro}, we briefly explain the main ideas behind the paper, such as $\epsilon$-nets, $t$-designs and heat kernels.

\item In Section~\ref{sec:results} we summarise the \textbf{main results} and their \textbf{applications}.

\item In Section~\ref{sec:connecting-pun-and-sun} we address step 1 of the proof.
\item In Section~\ref{sec:t-design-eps-net-bound}, we address steps 2 and 3 of the proof and combine them to prove a bound for the $t$ sufficient for a projective unitary $t$-design to be an $\epsilon$-net (Theorem \ref{th:theorem1}).
\item In Section~\ref{sec:delta-t-design-eps-net-bound} address step 4 of the proof and combine the bounds from steps 2-4 to derive the bounds on $t$ and $\delta$ sufficient for a projective unitary $\delta$-approximate $t$-design to be an $\epsilon$-net, realising step 5 of the proof (Theorem \ref{th:theorem2}).

\item In Section~\ref{sec:approx-delta}, we summarize the technical properties of trimmed heat kernels as approximations to the  Dirac delta (Theorem \ref{th:theorem3}).

\item Finally, in Section~\ref{sec:summary}, we provide a summary and outline the future research directions.

\item The appendix contains the proofs of various technical lemmas.
\end{itemize}

\section{Main ideas}
\label{sec:intro}

Central in quantum information theory is the concept of unitary channels. Such channels act via unitary operations (lossless quantum gates) when restricted to pure quantum states.

The unitary channel $\mathbf{U}$ acting on a Hilbert space $\mathcal{H} \cong \mathbb{C}^d$ is the CPTP map defined via $\mathbf{U}(\rho) = U \rho U^{\dagger}$, for any quantum state $\rho: \mathcal{H} \rightarrow \mathcal{H}$ and some fixed unitary $U \in \mathrm{U}(d)$. Since two unitaries $U$ and $V$ which differ by a phase $U=Ve^{i \phi}$ define the same unitary channel, the set of all unitary channels can be identified with the projective unitary group $\mathrm{PU}(d)=\mathrm{U}(d)/\mathcal{Z}(\mathrm{PU}(d))$, where $ \mathcal{Z}(\mathrm{PU}(d))=\{e^{i\phi}I,\, \phi\in (-\pi, \pi]\} \cong \mathrm{U}(1)$ is the centre of $\mathrm{U}(d)$.

Since we prefer to work with the $\mathrm{SU}(d)$ group, in our considerations we assume $U\in \mathrm{SU}(d)$ and use $\mathrm{PU}(d)=\mathrm{SU}(d)/\mathcal{Z}(\mathrm{SU}(d))$, where $\mathcal{Z}(\mathrm{SU}(d))=\{e^{i \frac{2 \pi}{ d}k}I,\, k\in \mathbb{Z}\}\cong \mathbb{Z}_d$ is the centre of $\mathrm{SU}(d)$ (group of $d^\text{th}$ roots of unity). From now on we denote $\Gamma \coloneqq \mathcal{Z}(\mathrm{SU}(d))$ and use square brackets to denote the elements of the projective group as equivalence classes of elements of $\mathrm{SU}(d)$, i.e. $U$ is mapped to the unitary channel $\mathbf{U}$ under the quotient map $\pi: \mathrm{SU}(d) \rightarrow \mathrm{PU}(d)$.

In practice, one is often interested in the closeness of different unitary channels. Various norms (and induced metrics) can be used to quantify this. A prominent example is the diamond norm $||\cdot||_{\Diamond}$, which has a clear operational meaning in terms of the statistical distinguishability of two channels (e.g. determines the maximal probability of success in a single-shot channel discrimination task). We denote the induced metric as
$d_{\Diamond}\left(\mathbf{U},\mathbf{V}\right)=||\mathbf{U}-\mathbf{V}||_{\Diamond}$.



We define $d(\cdot, \cdot)$ to be a metric on $\mathrm{SU}(d)$ induced by the operator norm

\begin{align}
    d(U, V) & \coloneqq \norm{U -  V}_\infty.
\end{align}
Since we want to work with the group $\mathrm{SU}(d)$, we define the metric on $\mathrm{PU}(d)$ in terms of the former
\begin{align}
    d_P(\mathbf{U}, \mathbf{V}) &\coloneqq \min_{\gamma\in\Gamma} d(U, \gamma V) \mathrm{.}
\end{align}
Clearly, due to the unitary invariance of the operator norm, the metrics $d(\cdot, \cdot)$ and $d_P(\cdot, \cdot)$ are translation-invariant.

One may show \cite{9614165} that $d_{\Diamond}(\cdot, \cdot)$ and $d_P(\cdot, \cdot)$ are related as

\begin{equation}
    d_P(\mathbf{U},\mathbf{V}) \leq d_{\Diamond}(\mathbf{U},\mathbf{V}) \leq 2 \, d_P(\mathbf{U},\mathbf{V}) \mathrm{.}
\end{equation} 

We say that a finite subset of channels $\mathcal{A} \subset \mathrm{PU}(d)$ is an $\epsilon$-net if for every channel $\mathbf{U} \in \mathrm{PU}(d)$, there exists a channel $\mathbf{V} \in \mathcal{A}$, such that $d_P(\mathbf{U}, \mathbf{V}) \leq \epsilon$. In other words, $\mathcal{A}$ represents all the possible channels, up to the error $\epsilon$.

To consider unitary designs, we need to define integration of functions on $\mathrm{PU}(d)$. The Haar measure $\mu_P$ on $\mathrm{PU}(d)$ is the pushforward of the Haar measure $\mu_S$ on $\mathrm{SU}(d)$, i.e. $\mu_P(\mathcal{A})=\mu_S(\pi^{-1}(\mathcal{A}))$, whenever $\pi^{-1}(\mathcal{A})$ is $\mu_S$-measurable.

Every function $f$ on $\mathrm{PU}(d)$ can be lifted to a unique function $\tilde{f}$ on $\mathrm{SU}(d)$, so that 
$\tilde{f}(U)=f(\mathbf{U})$. Clearly, such a function is constant on the equivalence classes (fibres of $\pi$), i.e. all the elements $U$ that define the same unitary channel. Conversely, every function $\tilde{f}$ on $\mathrm{SU}(d)$ which is constant on the equivalence classes, descends to a unique function $f$ on $\mathrm{PU}(d)$, so that
$\tilde{f}(U)=f(\mathbf{U})$. Hence, we can write
\begin{align}
    \int_{\mathrm{PU}(d)} f\,d\mu_P &= \int_{\mathrm{SU}(d)} \tilde{f}\,d\mu_S.\label{eqn:pu-integration-su-integration2}
\end{align}
If $X\subseteq \mathrm{PU}(d)$ is some Haar-measurable set then inserting indicator functions into~\eqref{eqn:pu-integration-su-integration2} we obtain
\begin{align}
    \int_{X} f\,d\mu_P &= \int_{\tilde{X}} \tilde{f}\,d\mu_S \mathrm{,}
\end{align}
where $\tilde{X}=\pi^{-1}(X)$.
This allows us to move freely between the $\mathrm{PU}(d)$ and $\mathrm{SU}(d)$ settings.

The (unitary) $t$-design on $\mathrm{PU}(d)$ is the probability measure $\nu$ on $\mathrm{PU}(d)$ which mimics the averaging properties of the Haar-measure with respect to the polynomials of degree at most $t$. Specifically, let $\mathcal{H}_t$ denote the space of homogeneous polynomials of degree $t$ in the matrix elements of $U$ and in $\overline{U}$.

A probability measure $\nu$ on $G$ is a unitary $t$-design if for any $f\in \mathcal{H}_t$ we have

\begin{equation}
    \int_{G}d\nu(U)f(U)=\int_{G}d\mu(U)f(U) \mathrm{.}
    \label{eq:tde}
\end{equation}

From the practical point of view, one is often interested in the case of $\nu$ being a discrete finitely supported measure, so that the averaging takes place over a finite set of elements $\{\nu_i, U_i\}$

\begin{equation}
    \sum_i \nu_i f(U_i) = \int_{G}d\mu(U)f(U)  \mathrm{.}
    \label{eq:discr_du}
\end{equation}

For example, $\nu$ can be the probability measure supported on a finite universal set of quantum gates $\mathcal{S}=\{U_1, U_2, \ldots, U_k\}$, which we denote as $\nu_{\mathcal{S}}$. In this work, we assume all $t$-designs have finite support since we directly apply Lemma~2 of Ref.~\cite{9614165}. However, this Lemma can be generalised to infinitely supported or even continuous measures, so that our results hold for such cases as well, if the definition of the $\epsilon$-net is relaxed by removing the finiteness condition. 

Moreover, it is useful to consider the cases in which~\eqref{eq:tde} is satisfied only approximately. To do so, it is useful to define so-called $t$-moment operators

\begin{equation}
     T_{\mu,t}\coloneqq\int_{G} d\mu(U) U^{ t,t},\,\,\,\,\,T_{\nu,t}\coloneqq \int_G d\nu(U) U^{ t,t}  \mathrm{.}
\end{equation}

One may check that the space $\mathcal{H}_t$ is spanned by the entries of
$U^{t,t}:=U^{\otimes t}\otimes \bar{U}^{\otimes t}$. Indeed for every $f \in \mathcal{H}_t$ there exists a matrix $A$ such that $f(U)=\mathrm{Tr}(A U^{t,t})$.

This way, the deviation from $\nu$ being a $t$-design~\eqref{eq:tde} can be measured as (see Ref.~\cite{9614165})
\begin{equation}
   \delta(\nu,t):=\left\|T_{\nu,t}-T_{\mu,t}\right\|_{\infty}\in [0,1] \mathrm{,}
   \label{eq:delta-metric}
\end{equation}
where $\|\cdot\|_{\infty}$ is an operator norm, leading to the notion of $\delta$-approximate $t$-designs, for which $\delta(\nu,t) < 1$ and exact $t$-designs, for which $\delta(\nu,t)=0$.

Finally, we recall that for $s < t$ we have $\mathcal{H}_s \subset \mathcal{H}_t$, hence $t$-designs are also $s$-designs.


The techniques used in this paper are similar to the ones from Ref.~\cite{9614165} and include the usage of the approximations to the  Dirac delta on compact groups. However, in this paper, we employ approximations based on the heat kernel - a natural and well-known object, contrary to the periodised Gaussian construction from Ref.~\cite{9614165}. We will first introduce the heat kernel with an elementary classical example.

\begin{exa}[Heat equation on a circle and the Poisson summation formula]
\label{ex:circle}
Consider a circle $S^1 \cong \mathbb{R}/\mathbb{Z}$ as an example of a 1-dimensional flat torus. Denoting the coordinate as $\phi$, the metric tensor induced from the Euclidean metric on $\mathbb{R}$ is $g= d\phi^2$. Hence $\Delta = \frac{\partial^2}{\partial \phi^2} $ and the heat equation \footnote{Physically we consider heat equations with unit conductivity.} on such a manifold reads $ u_t(t,\phi)=u_{\phi\phi}(t,\phi)$ with the initial condition $u(0,\phi)=f(\phi)$. We assume that $f(\phi)$ is square integrable, i.e. $f(\phi) \in L^2(S^1)$. Such a problem is typically solved by considering the corresponding equation on $\mathbb{R}$ with the periodic boundary conditions, separation of variables and the expansion of the initial data $f(\phi)$ to the Fourier series. Here, we take a different approach - we find the fundamental solution to the corresponding problem on $\mathbb{R}$ and periodise it. We use the definition of the Fourier transform of a (complex) function in $L^2(\mathbb{R})$ as the unique unitary extension of the map $g \mapsto \hat{g}$, $g \in L^1(\mathbb{R}) \cap L^2(\mathbb{R})$, where
\begin{equation}
   \hat{g}(\xi) = \int_{-\infty}^{\infty}   e^{-i2 \pi \xi x} g(x) dx, \quad \forall \xi \in \mathbb{R} \mathrm{.}
\end{equation}
From now on, we fix $t$ and consider the corresponding single-variable functions on $\mathbb{R}$. We unwrap the initial datum $f$ into an interval $[0, 1) \subset \mathbb{R}$. We denote the corresponding functions on $\mathbb{R}$ using the same symbols as for $S^1$. Applying the Fourier transform, we obtain
\begin{equation}
\label{eq:he_fourier}
    \hat{u}_t(\xi,t)+4 \pi^2\xi^2 \hat{u}(\xi,t)=0
\end{equation}
with the initial condition $\hat{u}(\xi,0)=\hat{f}(\xi)$, where $4 \pi^2\xi^2$ is the eigenvalue of $-\Delta$.
Multiplying~\eqref{eq:he_fourier} by $e^{4 \pi^2\xi^2t}$ we obtain $\frac{\partial}{\partial t}\left(  e^{4 \pi^2\xi^2t}\hat{u}(\xi,t)\right)=0$. Hence $e^{4 \pi^2\xi^2t}\hat{u}(\xi,t)$ is some function of $\xi$ and from the initial condition we see that $\hat{u}(\xi,t)=e^{-4\pi^2\xi^2t} \hat{f}(\xi)$. Denoting the inverse Fourier transform of $e^{-4\pi^2\xi^2t}$ as $H_{\mathbb{R}}(\phi,t)$ we obtain
\begin{equation}
    H_{\mathbb{R}}(\phi,t)=\frac{1}{\sqrt{4 \pi t}}e^{-\frac{\phi^2}{4t}} \mathrm{.}\label{eq:heat-kernel-real-line}
\end{equation}
Thus, the solution on $\mathbb{R}$ is the convolution (with respect to the $\phi$ variable)
\begin{equation}
    u(\phi,t)= \left(H_{\mathbb{R}}(\cdot,t) \ast f\right)(\phi) = \int_{-\infty}^\infty K_{\mathbb{R}}(\phi,\phi^\prime,t) f(\phi^\prime) d\phi^\prime \mathrm{,}
\end{equation}
where 
\begin{equation}
    K_{\mathbb{R}}(\phi,\phi^\prime,t)=H_{\mathbb{R}}(\phi-\phi^\prime,t)=\frac{1}{\sqrt{4 \pi t}}e^{-\frac{(\phi-\phi^{\prime})^2}{4t}} \mathrm{,}
\end{equation}
and $u(\phi, t)$ is smooth for all $t>0$.
To find the fundamental solution on $S^1$ we periodise $H_{\mathbb{R}}(\phi,t)$ obtaining a 1-periodic function on $\mathbb{R}$ and an equivalent function on $\mathbb{R}/\mathbb{Z}$
\begin{equation}
\label{eq:H_S1}
    H_{S^1}(\phi,t) = \frac{1}{\sqrt{4 \pi t}} \sum_{n\in \mathbb{Z}} e^{-\frac{(\phi+n)^2}{4t}}, \quad \phi \in \mathbb{R}/\mathbb{Z} \mathrm{.}
\end{equation}
To rewrite~\eqref{eq:H_S1} we can use the Poisson summation formula, which states that for a complex-valued function $s(x)$ on $\mathbb{R}$ whose all derivatives decay at infinity (i.e. a Schwartz function)

\begin{equation}
    \sum_{n=-\infty}^{\infty} s(n) =  \sum_{k=-\infty}^{\infty} \hat{s}(k) \mathrm{.}
\end{equation}
Treating~\eqref{eq:H_S1} as a 1-periodic function on $\mathbb{R}$ we apply the Poisson summation formula with $s(x)= e^{-\frac{(\phi+x)^2}{4t}}$ and obtain
\begin{equation}
    H_{S^1}(\phi,t) =  \sum_{k\in \mathbb{Z}} e^{-i 2 \pi k \phi} e^{-(2\pi k)^2 t}, \quad \phi \in \mathbb{R}/\mathbb{Z} \mathrm{,}
\end{equation}
which is the complex Fourier series expansion. Rewriting it into the sine-cosine form yields
\begin{equation}
    H_{S^1}(\phi,t) = 1 + 2 \sum_{k=1}^{\infty} \mathrm{cos}(2\pi k \phi) e^{-(2\pi k)^2 t}    \quad \phi \in \mathbb{R}/\mathbb{Z} \mathrm{,}
\end{equation}
which is a linear combination of eigenfunctions $2 \mathrm{cos}(2 \pi k \phi)$ with eigenvalues $-4 \pi^2 k^2$ and is of the same form as the fundamental solution obtained via the typical Fourier series expansion approach.
\end{exa}
Generalizing the heat kernel on $\mathbb{R}$ shown in equation~\eqref{eq:heat-kernel-real-line}, the heat kernel on $\mathbb{R}^d$ has the form
\begin{equation}
\label{heat-kernel-RN}
    K(t,x,y) = \frac{1}{(4 \pi t)^{d/2}} e^{-||x-y||^2/4t} \mathrm{,}
\end{equation}
where $||\cdot||$ is the Euclidean norm, defined for any $x,y \in \mathbb{R}^d$ and $t>0$. This is the fundamental solution to the heat equation
\begin{equation}
\label{eq:heat_kernel_on_R_N}
   u_t(t,x)=\Delta u(t,x) \mathrm{,}
\end{equation}
where $\Delta$ is the Laplacian on $\mathbb{R}^d$. One can consider the generalisation of the heat equation~\eqref{eq:heat_kernel_on_R_N} to other spaces, e.g. Riemannian manifolds $(M,g)$, by replacing $\Delta$  with the Laplace-Beltrami operator (in local coordinates)

\begin{equation}
    \Delta f = \frac{1}{\sqrt{|g|}} \partial_i  \left( \sqrt{|g|} g^{ij} \partial_j f\right)  \mathrm{,}
\end{equation}
acting on differentiable functions $f$ on $M$.

Following the method demonstrated in Example~\ref{ex:circle}, we can derive the heat kernel for the flat $d$-dimensional torus $\mathbb{R}^d/\Lambda$
by taking the heat kernel on $\mathbb{R}^d$ shown in equation~\eqref{heat-kernel-RN} and periodising the solution over the lattice $\Lambda \cong \mathbb{Z}^d$. The Poisson formula also generalises to higher dimensions. A thorough introduction to this topic may be found in Ref.~\cite{Maher-thesis}.

However, in our work, we are interested in heat kernels on Lie groups. To make sense of the heat equation on a Lie group, the proper Riemannian structure needs to be chosen. For compact semi-simple Lie group $G$, the Riemannian structure $(G, g)$ can be defined via Ad-invariant positive definite inner product $(\cdot , \cdot)$ stemming from the Killing form~\footnote{Taking the negative of the negative-definite Killing form leads to the positive-definite scalar product.}.

Notice that although the group $\mathrm{U}(d)$ is compact, it is not semi-simple. Hence, the metric tensor stemming from the Killing form is only positive semi-definite. Indeed, one can check that such a metric tensor for $\mathrm{U}(1) \cong S^1$ is identically zero, so it does not equip $S^1$ with the Riemannian structure. This is in contrast with the construction from Example \ref{ex:circle}.

Of course, general Lie groups are not commutative. Hence, in order to study the heat equation on a compact Lie group $G$, non-commutative Fourier/harmonic analysis is needed. Fourier coefficients on a compact Lie group are calculated with respect to the irreducible representations (irreps) of the group. Generally, the object being transformed is the regular Borel measure on $G$. However, we focus on the related case of integrable functions $f$. In this case (see e.g.~\cite{faraut_2008}), the Fourier coefficient $\hat{f}(\lambda)$ is the operator in $\mathrm{End}(V_{\pi_{\lambda}})$ defined via
\begin{equation}
\label{eq:fourier-nonab}
    \hat{f}(\lambda)=\int_G  \pi_{\lambda}(g^{-1}) f(g) d\mu(g)  \mathrm{,}
\end{equation}
where by $V_{\pi_{\lambda}}$ we denote the representation space of irrep $\pi_{\lambda}$ with highest weight $\lambda$.
Equipping the space $\mathrm{End}(V_{\pi_{\lambda}})$ with the norm $\sqrt{d_{\lambda}}||\cdot||_{HS}$, where $d_{\lambda} \coloneqq \mathrm{dim} \left(V_{\pi_{\lambda}}\right)$ and the Hilbert-Schmidt norm $||u||^2_{HS}=\mathrm{Tr}(uu^{*})$, one can show that such the Fourier transform is an isomorphism of Hilbert spaces 

\begin{equation}
     L^2(G) \cong \bigoplus_{\pi \in \hat{G}} \mathrm{End}(V_{\pi_{\lambda}}) \mathrm{,}
\end{equation}
where $\hat{G}$ is the set of equivalence classes of irreps of $G$. Namely, we obtain a generalisation of the Plancherel's theorem
\begin{equation}
    ||f||^2_2= \int_G |f(g)|^2 d\mu(g) = \sum_{\lambda \in \hat{G}} d_{\lambda} ||\hat{f}(\lambda)||_{HS}^2 \mathrm{.}
    \end{equation}
This is a consequence of the Peter-Weyl theorem.

\begin{rem}
    The transform~\eqref{eq:fourier-nonab} is a generalization of the Fourier series. Indeed, suppose a compact group $G$ is additionally abelian and connected (so is a torus). Take one-dimensional torus $\mathrm{U}(1) \cong S^1$ for example. The unitary irreps $\pi_{\lambda}$ of $\mathrm{U}(1)$ are the homomorphisms $\mathrm{U}(1) \rightarrow \mathrm{U}(1)$ so they are of the form $e^{i\phi} \mapsto  e^{i \lambda \phi}$ for some integer $\lambda$. All irreps are one-dimensional and $\hat{S^1} \cong \mathbb{Z}$. The Fourier coefficients of a function $f: \mathrm{U}(1) \rightarrow \mathbb{C}$ are
\begin{equation}
    \hat{f}(\lambda)=\frac{1}{2 \pi} \int_{-\pi}^{\pi} e^{-i \lambda \phi} f(e^{i\phi}) d \phi \mathrm{,}
\end{equation}
which coincides with the Fourier coefficients of the corresponding $2\pi$-periodic complex-valued function $ \tilde{f}: \mathbb{R} \rightarrow \mathbb{C}$, $\tilde{f}(x)=f(e^{i x})$.
Similarly, other results such as the completeness, orthogonality relations and Plancherel's theorem generalise to the non-abelian case via the Peter-Weyl theorems and representation theory.
\end{rem}

Heat kernels on simply-connected compact semi-simple Lie groups were studied in Ref.~\cite{urakawa-1974}, together with a useful Poisson form. In Section \ref{sec:connecting-pun-and-sun} we show how to apply those results for $\mathrm{PU}(d)$, which is not simply-connected.

\section{Main results and applications}
\label{sec:results}
Below we summarise the main results of this paper and outline some of their applications.
\begin{res}
     The main technical result of the paper is the construction of the polynomial approximation to the Dirac delta function $H_P^{(t)}(\cdot, \sigma)$ on $\mathrm{PU}(d)$, together with some of its properties.
    This allows us to summarise the key properties of the family of polynomial approximations of Dirac delta based on the trimmed heat kernels (see Theorem \ref{th:theorem3} for a precise statement).
\end{res}

\begin{res}
     The supports of exact $t$-designs in $\mathrm{PU}(d)$ with $d\geq 2$ are $\epsilon$-nets for $t \simeq \frac{d^{\frac{5}{2}}}{\epsilon}$ (see Theorem \ref{th:theorem1} for a precise statement).
\end{res}

\begin{res}
\label{res:main}
The supports of approximate $t$-designs in $\mathrm{PU}(d)$ with $d\geq 2$ are $\epsilon$-nets for $t \simeq \frac{d^{\frac{5}{2}}}{\epsilon}$ and $\delta \simeq \left(\frac{\epsilon}{d^{1/2}}\right)^{d^2}$.
(see Theorem \ref{th:theorem2} and its proof for a precise statement). This provides an ``essentially yes'' answer to the conjecture about the optimal scaling of $t(\epsilon,d)$ from Section IV in Ref.~\cite{9614165}.
\end{res}

\begin{app}[Efficiency of quantum gates]
\label{appl:1}

    This result is analogous to Proposition 2 from Ref.~\cite{9614165} and is a simple consequence of Result \ref{res:main}. For example, using the bound~\eqref{eqn:delta-bound-with-kappa} from the proof of Theorem \ref{th:theorem2}, one may prove that if $\nu$ is a discrete probability measure on $\mathrm{PU}(d)$ with $d \geq 2$, which is a $\delta$-approximate $t$-design with $\delta=\delta(\nu,t)$ for
    \begin{equation}
    \label{eq:t_cond}
        t \geq  32\frac{d^{\frac{5}{2}}}{\epsilon}\log(d)\log\left(\frac{4}{a_v \epsilon}\right) \mathrm{,}
    \end{equation}
    where $C=9 \pi$,
    then the support of $\nu^{*\ell}$ forms an $\epsilon$-net in $\mathrm{PU}(d)$ for
    \begin{equation}
        \ell \geq \frac{\mathrm{log}(1/\kappa(d))+(d^2-1)\left(\frac{5}{4}\mathrm{log}\left(\frac{1}{\epsilon}\right)+\frac{3}{4}\mathrm{log}(Dd)\right)}{\mathrm{log}\left(1/\delta(\nu,t)\right)} \mathrm{,}
    \end{equation}
    where
    \begin{equation}
        D =8C^{2/3} \mathrm{log}^{1/3}\left(2 C\right)  \mathrm{,}
    \end{equation}
    and $\mathrm{log}(1/\kappa(d)) < 5$. Moreover $\mathrm{log}(1/\kappa(d))< 0$ for $d \geq 9$. Hence, in the case of the measure $\nu_{\mathcal{S}}$, the support of $\nu_{\mathcal{S}}^{*\ell}$ are the length $\ell$ words built out of the elements of $\mathcal{S}$ and this result is the SKL theorem with $\mathrm{log}(\frac{1}{\epsilon})$ term but also the multiplicative factor $\mathrm{log}^{-1}(1/\delta(\nu,t))$, which depends on $t$ (or $\epsilon$ e.g. by taking~\eqref{eq:t_cond} as equality). Such SKL theorems can be used to bound the overhead of quantum circuits \cite{Slowik_qco_2024}. 
\end{app}

\begin{app}[Inverse-free SK theorem]
\label{appl:2}
Similarly as in Ref.~\cite{9614165}, Application \ref{appl:1} can be turned into the inverse-free non-constructive SKL theorem without the $\epsilon$-dependent multiplicative factor $\mathrm{log}^{-1}(1/\delta(\nu,t))$, by bounding the decay of $1-\delta(\nu,t)$ with growing $t$, using the results from Ref.~\cite{varjú2015randomwalkscompactgroups}. Namely, let $\nu_{\mathcal{S}}$ be a uniform probability measure on $\mathcal{S} \subset \mathrm{PU}(d)$. Then the support of $\nu_{\mathcal{S}}^{*\ell}$ is an $\epsilon$-net in $\mathrm{PU}(d)$ for

\begin{equation}\label{eqn:application2-with-group-constants}
    \ell \geq A \frac{\mathrm{log}^3\left(\frac{1}{\epsilon}\right)+B}{\mathrm{log}(1/\delta(\nu,t_0))}  \mathrm{,}
\end{equation}
where $A, B$ and $t_0$ are some positive group constants. However, the constants are unknown due to the ambiguity of constants presented in Ref.~\cite{varjú2015randomwalkscompactgroups}.
\end{app}

\begin{app}[Quantum complexity and black hole physics]

    This application comes from the Ref.~\cite{PhysRevX.14.041068} in which the authors use the approximation of Dirac delta construction from Ref.~\cite{9614165} to prove the results about the approximate equidistribution of $\delta$-approximate $t$-designs in the space $\mathbf{U}(d)$. This is then used to obtain results about the saturation and recurrence of the complexity of random local quantum circuits with gate set $\mathcal{S}$ without the assumptions on the spectral gap or inverse-closeness of $\mathcal{S}$. Such circuits can be used to model the chaotic dynamics of quantum many-body systems, which may be applicable in areas such as the physics of black hole interiors.
    
    We believe that after some work, using our construction, one may obtain the approximate equidistribution of $\delta$-approximate $t$-designs (Theorem 16 from Ref.~\cite{PhysRevX.14.041068}) with better scaling in $\epsilon$ and $d$, which translates to the saturation and recurrence results. 
\end{app}

\section{The heat kernel on the projective unitary group}
\label{sec:connecting-pun-and-sun}

In the sequel, we employ formulae which are known for the heat kernel on $\mathrm{SU}(d)$, but which do not appear to be readily available for that on $\mathrm{PU}(d)$. Using standard techniques, we are able to write the latter in terms of the former in order to generalise the formulae we need.

Before we do so, we recall some facts from the representation theory of Lie groups (see e.g. Ref.~\cite{Ha15, FuHa91, Ki08}) and fix some notation and relevant conventions. 

We work over the field of complex numbers. Let $K$ be a (real) compact simply-connected Lie group (e.g. $\mathrm{SU}(d)$). Due to compactness, we can restrict ourselves to unitary complex representations. The complex representation theory of $K$ is equivalent to the complex representation theory of its Lie algebra $\mathfrak{k}$, which is equivalent to the complex representation theory of its complexification $\mathfrak{g} = \mathfrak{k} + i\mathfrak{k}$. 

The Cartan subalgebra of Lie algebra $\mathfrak{g}$ is an abelian and diagonalisable subalgebra of $\mathfrak{g}$, which is maximal under set inclusion. In general, there are many ways to choose the Cartan subalgebra. In our case, we can fix it by choosing the maximal torus in the Lie group. Let $T$ be the maximal torus in $K$ with Lie algebra $\mathfrak{t}$. Then the corresponding Cartan subalgebra $\mathfrak{h}$ of $\mathfrak{g}$ is $\mathfrak{h}= \mathfrak{t}+i\mathfrak{t}$.

For $K=\mathrm{SU}(d)$, we have $\mathfrak{g}=\mathfrak{sl}(d, \mathbb{C})$, which is a finite-dimensional complex semi-simple Lie algebra. The theory of finite-dimensional complex representations of such algebras is well-known and particularly nice. For example, such algebras are classified by their root systems/Dynkin diagrams and such representations are characterised by the theorem of the highest weight. Here, to match the notation of  Ref.~\cite{urakawa-1974}, we take a slightly different approach than usual, which is more suitable for compact groups $K$. In particular, we consider real weights and roots. 

Let $(\Pi, V)$ be a (finite-dimensional) representation of $K$ and $\pi$ be the associated representation of $\mathfrak{g}$. The (real) weight of $V$ with respect to $\mathfrak{t}$ is an element $\lambda$ from the dual space $\mathfrak{t}^*$, such that the corresponding weight space

\begin{equation}
    V_{\lambda} \coloneqq \{v \in V| \quad  \pi(H)v=i \lambda(H)v, \forall H \in \mathfrak{t}\}
\end{equation}
is not zero. Hence, the (real) root of $\mathfrak{g}$ with respect to $\mathfrak{t}$ is the non-zero element $\alpha$ from the dual space $\mathfrak{t}^*$, such that the corresponding root space
\begin{equation}
   \mathfrak{g}_{\alpha}:= \{E \in \mathfrak{g}| \quad  [H,E]=i \alpha(H)E, \forall H \in \mathfrak{t}\} 
\end{equation}
is not zero. We denote the root system of $\mathfrak{g}$ as $\Phi$, the set of all positive roots as $\Phi^+$ and the set of simple roots as $\Delta$.

Additionally, we assume that $K$ is simply-connected. The algebra $\mathfrak{k}$ is equipped with $Ad(K)$-invariant positive-definite inner product $(\cdot, \cdot)$ defined as the negative of its Killing form (which is non-degenerate and negative-definite)

\begin{equation}
\label{eq:killing_tensor}
    (X,Y) \coloneqq - \mathrm{Tr} \left( \mathrm{ad}(X) \circ \mathrm{ad}(Y) \right) \mathrm{.}
\end{equation}
The restriction of $(\cdot, \cdot)$ to $\mathfrak{t}$ is non-degenerate (hence, it defines the inner product on $\mathfrak{t}$). Thus, can use $(\cdot, \cdot)$ to identify $\mathfrak{t} \cong \mathfrak{t}^{*}$ via $X \mapsto \lambda_X$ for $X \in \mathfrak{t}$, where $\lambda_X(Y)=(X,Y)$ for any $Y \in \mathfrak{t}$ and $\lambda \mapsto X_{\lambda}$ for $\lambda \in \mathfrak{t}^{*}$, where $\lambda(Y)=(X_\lambda,Y)$ for any $Y\in \mathfrak{t}$. This way we also define $(\cdot, \cdot)$ on $\mathfrak{t}^{*}$ as $(\lambda,\kappa)=(X_{\lambda}, X_{\kappa})$ for $\lambda, \kappa \in \mathfrak{t}^{*}$ and the induced norm $||\cdot||$. The inner product (\ref{eq:killing_tensor}) defines the Riemannian metric on $K$, hence also the Laplace-Beltrami operator $\Delta$. Thus, we can study the corresponding heat kernels.

Additionally for $\lambda \in \mathfrak{t}^*, \lambda \neq 0$ we define

\begin{equation}
\label{eq:coroot}
    \lambda^* \coloneqq \frac{2}{(\lambda, \lambda) } \lambda 
\end{equation}
and the Weyl vector
\begin{equation}
\label{eqn:weyl-vector-delta-defn}
    \delta \coloneqq \frac{1}{2} \sum_{\alpha \in \Phi^+ } \alpha \mathrm{.}
\end{equation}

We aim to describe the heat kernel on $\mathrm{PU}(d)$ in terms of the heat kernel on $\mathrm{SU}(d)$. Specialising to the case $K = \mathrm{SU}(d)$, we introduce
\begin{align}
    \Gamma = \left\{\exp\left(\frac{2k\pi}{d}\right) I \,\middle|\, k \in \mathbb{Z}\right\} \cong \mathbb{Z}_d,\label{eqn:SU(D)-center}
\end{align}
so that $K/\Gamma \cong \mathrm{PU}(d)$. Our approach is based on the averaging map

\begin{equation}
    f(x) \mapsto \frac{1}{|\Gamma|}\sum_{\gamma \in \Gamma} f(\gamma x) \mathrm{.}
\end{equation}
Every irrep of $\mathrm{PU}(d)$ extends to an irrep of $\mathrm{SU}(d)$ by making it constant on $\Gamma$-cosets. It follows from Lemma~\ref{lem:averaging-irreps-is-nice} in Appendix \ref{app:heat_kenrel_on_su_and_pu} that every irrep of $\mathrm{PU}(d)$ is obtained by applying the averaging map to a corresponding irrep of $\mathrm{SU}(d)$. Let's consider an elementary example.

\begin{exa}[Irreps of $\mathrm{SU}(2)$ and $\mathrm{SO}(3)$]
    In this example we consider the irreps of $\mathrm{SU}(2)$ and aim to find the corresponding representations of $\mathrm{PU}(2)\cong \mathrm{PSU}(2)\cong \mathrm{SO}(3)$.
    
    The irreps of $\mathrm{SU}(2)$ can be enumerated by the corresponding particle spin $j=0,\, \frac{1}{2},\, 1, \ldots$ and have dimensions $2j+1$ (i.e. single irrep in each dimension). The centre $\Gamma \cong \mathbb{Z}_2$ acts by $\pi$-shifts. We want to find the irrep obtained via the averaging map applied to the irrep with spin $j$. We start with the spin $j$ character of $\mathrm{SU}(2)$
    \begin{equation}
    \label{eq:char}
\chi_j(\theta)=\frac{\mathrm{sin}\left((2j+1)\theta\right)}{\mathrm{sin}(\theta)} \mathrm{.}
    \end{equation}
    Averaging~\eqref{eq:char} yields
    \begin{equation}
        \frac{1}{2}\left(\frac{\mathrm{sin}\left((2j+1)\theta\right)-\mathrm{sin}\left((2j+1)\theta+(2j+1)\pi\right)}{\mathrm{sin}(\theta)}\right)= \begin{cases}
\chi_j(\theta) , \quad &\text{for j being full-integer,}\\
0, \quad &\text{for j being half-integer.}
\end{cases}
    \end{equation}
Hence, we obtain a well-known fact that the full-integer spin irreps of $\mathrm{SU}(2)$ are projective.
\end{exa}

Thus, we can focus on the description of the heat kernel on $\mathrm{SU}(d)$. The first formula we employ is the standard expression for the heat kernel as the combination of characters, valid for compact semi-simple simply-connected Lie groups 
\begin{align}
\label{eq:heat_kernel_original}
    H_{S}(g, \sigma) &= \sum_\lambda d_\lambda \exp\left(-k_\lambda\sigma\right)\chi_\lambda( g),
\end{align}
where $\lambda$ is the highest weight vector and the sum is over complex irreps, $d_\lambda$ is the dimension of the irrep, $\chi_{\lambda}$ is the character and $k_\lambda \coloneqq (\lambda + 2\delta, \lambda)$ - see Ref.~\cite{faraut_2008,urakawa-1974,ojm/1200693127}. The parameter $\sigma > 0$ plays the role of time and the subscript $S$ indicates that this is the heat kernel on $\mathrm{SU}(d)$; later we will use $H_P$ to denote the equivalent object on $\mathrm{PU}(d)$. The formula (\ref{eq:heat_kernel_original}) is, in fact, the decomposition in terms of the eigenfunctions of the Laplace-Beltrami operator $\Delta$, which are the characters $\chi_{\lambda}$, where

\begin{equation}
    \Delta \chi_{\lambda} = - k_{\lambda} \chi_{\lambda} \mathrm{.}
\end{equation}

In order to describe the highest weights $\lambda$ for $\mathrm{SU}(d)$ using vectors, we introduce the linear functionals on $\mathfrak{t}$ (see~\eqref{eqn:maximal-torus-parameterisation}) acting as
\begin{align}
    L_j : \begin{pmatrix}
        i \phi_1\\&i \phi_2\\&&\ddots\\&&& i \phi_d
    \end{pmatrix} \mapsto  \phi_j \mathrm{,}
\end{align}
so that $\lambda=\sum_{i=1}^d \lambda_i L_i $. Then the highest weights of $\mathrm{U}(d)$ can be labelled by integer-valued vectors $(\lambda_1, \lambda_2, \ldots, \lambda_d)$ with non-increasing entries, i.e. $\lambda_i \geq \lambda_{i+1}$ for $1 \leq i \leq d-1$. One can show that any irreducible representation of $\mathrm{U}(d)$ restricts to an irreducible representation of $\mathrm{SU}(d)$, while any irreducible representation of $\mathrm{SU}(d)$ extends to one of $\mathrm{U}(d)$. However, this mapping is not one-to-one. Since $\sum_{i=1}^d L_i (x)= 0$ for any $x\in\mathfrak{sl}(d,\mathbb{C})$ any irreducible representations of $\mathrm{U}(d)$ labelled by vectors which differ by a constant vector $(n, n, \ldots, n)$ for some $n\in \mathbb{Z}$ correspond to the same irreducible representation of $\mathrm{SU}(d)$.

We will also need to consider the irreducible representations of $\mathrm{PU}(d)$, which consists of equivalence classes of members of $\mathrm{U}(d)$ under the equivalence relation $U\sim e^{i\phi} U$. Any irrep of $\mathrm{PU}(d)$ extends to an irrep of $\mathrm{U}(d)$ by choosing it to be constant on equivalence classes so we can again label irreps of $\mathrm{PU}(d)$ with the same labels as those of $\mathrm{U}(d)$. An irrep of $\mathrm{U}(d)$ corresponds to an irrep of $\mathrm{PU}(d)$ exactly when it is constant on equivalence classes, which happens when the highest weight vector satisfies $\sum_j \lambda_j = 0$.
We denote \footnote{Not to be confused with the norm $||\cdot||$ stemming from the Killing form}
\begin{equation}
   ||\lambda||_1 \coloneqq \sum_{i=1}^d |\lambda_i|. 
\end{equation}
By restricting to $||\lambda||_1 \leq 2t$, we obtain the set of vectors labelling the projective irreps corresponding to the $t$-design, i.e. appearing in the decomposition of the representation $U^{t,t}$ \cite{348523741eed4f2b81d45ac8169a1432}.
For $\mathrm{SU}(d)$ the dimension of the representation $d_\lambda$ and the eigenvalue $k_\lambda$  may be expressed as (see Ref.~\cite{faraut_2008})
\begin{align}
    d_\lambda &= \chi_\lambda(e) = \frac{\prod_{j<l} (\lambda_j - \lambda_l +l - j) }{\prod_{j<l} (l - j) } \leq\left(1 + \norm{\lambda}_1 \right)^{d(d-1)/2}\label{eqn:irrep-dimension-bound}  \mathrm{,}\\
    k_\lambda &= \frac{1}{2d}\sum_j\left( \lambda_j^2 + (d - 2j+1)\lambda_j \right) - \frac{1}{2d^2} (\sum_j \lambda_j)^2.
\end{align}
If we have $\sum_j\lambda_j = 0$, so the $\mathrm{SU}(d)$ irrep is also a $\mathrm{PU}(d)$ irrep then we have the bound
\begin{align}
    k_\lambda \geq  \frac{\norm{\lambda}_1^2}{2d^2} + \frac{1}{4}\norm{\lambda}_1 \mathrm{.}\label{eqn:irrep-casimir-bound}
\end{align}
Since $\mathrm{SU}(d)$ and $\mathrm{PU}(d)$ share a Lie algebra and $d_{\lambda}$ with $k_{\lambda}$ can be computed in terms of properties of the Lie algebra, these are identical for special and projective unitary representations. Therefore, the averaging map may be applied term-wise to the formula for the heat kernel on $\mathrm{SU}(d)$ to obtain the corresponding one for $\mathrm{PU}(d)$
\begin{align}
    \frac{1}{\abs{\Gamma}}\sum_{\gamma\in \Gamma} H_\mathrm{S}(\gamma g, \sigma) &= \frac{1}{d}\sum_{\gamma\in \Gamma} \sum_\lambda d_\lambda \exp\left(-k_\lambda\sigma\right)\chi_\lambda(\gamma g) \label{eqn:shifted-heat-kernels}\\
    &=  \sum_\lambda d_\lambda \exp\left(-k_\lambda\sigma\right)\frac{1}{d}\sum_{\gamma\in \Gamma}\chi_\lambda(\gamma g) \\
    &=  \sum_\lambda d_\lambda \exp\left(-k_\lambda\sigma\right)\delta_\mathrm{P}(\lambda)\chi_\lambda(g) \\
    &=  H_\mathrm{P}( g, \sigma),\label{eqn:heat-kernel-character-formula}
\end{align}
where $H_\mathrm{S}$ and $H_\mathrm{P}$ are the special and projective unitary heat kernels, respectively, and from equation~\eqref{eqn:SU(D)-center}, we have $\abs{\Gamma} = d$. Here $\delta_\mathrm{P}$ is a Kronecker-delta like function, taking value $1$ for irreps of $\mathrm{SU}(d)$ which are also irreps of $\mathrm{PU}(d)$ (i.e. projective representations) and value $0$ otherwise. In fact,
\begin{align}
    \delta_P(\lambda) &=\begin{cases}1, &\sum_j\lambda_j = 0\\0,&\text{ otherwise}.\end{cases}
\end{align}
Notice that the heat kernel~\eqref{eq:heat_kernel_original} is a class function, hence it can be defined instead on the maximal torus of $\mathrm{SU}(d)$. With a mild abuse of notation, we will not distinguish between the two descriptions.

Let us formulate a second formula for the heat kernel on compact, semi-simple, simply-connected Lie groups from Ref.~\cite{urakawa-1974}
\begin{align}
    j(\exp(X)) &=(2i)^m \prod_{\alpha\in\Phi^+} \sin\left(\frac{\alpha(X)}{2}\right) \mathrm{,}\label{eqn:urakawa-j}\\
    K(X,\sigma) &= \sum_{\gamma\in\Gamma} \pi(\lambda_X + \gamma) \exp\left(-\frac{1}{4\sigma} \norm{\lambda_X + \gamma}^2\right)\label{eqn:urakawa-K} \mathrm{,}\\
    H(\exp(X),\sigma) &= \frac{c}{\pi(\delta)} (2\pi)^{l+m} i^m j(\exp(X))^{-1}\exp\left(\norm{\delta}^2 \sigma\right) (4\pi \sigma)^{-N/2} K(X,\sigma)\label{eqn:urakawa-H},
\end{align}
where 
\begin{equation}
    \pi(\lambda) \coloneqq \prod_{\alpha\in\Phi^+} (\lambda, \alpha) \mathrm{,} \quad \lambda \in \mathfrak{t}^* \label{eqn:urakawa-pi}\mathrm{,}
\end{equation}
$m=\abs{\Phi^+}$, $N$ is the group dimension, $c$ is a (known) dimension-dependent group constant and $\Gamma$ is a lattice generated by $l= \mathrm{dim}(\mathfrak{t})$ elements $\alpha_j^*$ (see~\eqref{eq:coroot}) corresponding to the simple roots $\Delta=\{\alpha_1, \ldots, \alpha_l\}$
\begin{equation}
     \Gamma \coloneqq 2\pi \sum_{j=1}^l \mathbb{Z} \alpha_j^* \mathrm{.}
\end{equation}
Formally, the heat kernel given by~\eqref{eqn:urakawa-H} is only defined for the regular elements $X$ from $\mathfrak{t}$, i.e. the ones with distinct eigenvalues. However, the corresponding set of group elements for which formula~\eqref{eqn:urakawa-H} is not well defined is of Haar-measure zero.  The function defined by the formula~\eqref{eqn:urakawa-H} extends 
to a unique continuous function,
that defined by~\eqref{eq:heat_kernel_original} on the whole group. Hence, with a slight abuse of notation, we treat~\eqref{eqn:urakawa-H} as defined on the whole group, e.g. when integrating. More explicitly, one can see that the non-definedness of~\eqref{eqn:urakawa-H} at non-regular elements arises exactly from the factor of $j(\exp(X))^{-1}$ giving factors of $\sin(\alpha(X)/2)^{-1}$ as $\alpha(X)\to 0$. However, these apparently singular terms are balanced by terms linear in $\alpha(X)$ arising from the term $\pi(\lambda_X + \gamma)$. We will sketch this in more detail in Appendix~\ref{sec:well-defined-heat-kernel}.

The formula~\eqref{eqn:urakawa-H} is equivalent to~\eqref{eq:heat_kernel_original} via the Poisson summation formula and we refer to it as the Poisson form (of the heat kernel). The Poisson form is relevant to us exactly because of the factor of $\sigma^{-1}$ appearing in the exponent in~\eqref{eqn:urakawa-K}. Roughly speaking, this formula is useful for bounding the behaviour of the heat kernel when the $\sigma$ is small, while equation~\eqref{eq:heat_kernel_original} is useful when $\sigma$ is large. 

The maximal torus $T$ of $\mathrm{SU}(d)$ may be identified with the group of determinant $1$ diagonal matrices parametrised by a vector $\phi \in \mathbb{R}^{d-1}$ as
\begin{align}
    T(\phi) \coloneqq \begin{pmatrix}
        e^{i\phi_1}\\ &e^{i\phi_2}\\&&\ddots\\&&&e^{i\phi_{d-1}}\\&&&&e^{-i\sum_{j=1}^{d-1}\phi_j}
    \end{pmatrix} \mathrm{.}
\end{align}

The Lie algebra $\mathfrak{t}$ of $T$ consists of traceless diagonal purely imaginary matrices parametrised by $\phi \in \mathbb{R}^{d-1}$ as
\begin{align}\label{eqn:maximal-torus-parameterisation}
    X(\phi) \coloneqq i\begin{pmatrix}
        \phi_1\\&\phi_2\\&&\ddots\\&&&\phi_{d-1}\\&&&&-\sum_{j=1}^{d-1}\phi_j
    \end{pmatrix} \mathrm{.}
\end{align}
Clearly, one can restrict the parameters e.g. $\phi_i \in (-\pi, \pi]$ for $1 \leq i \leq d-1$.

The complexified Lie algebra of $K$ is $\mathfrak{g}=\mathfrak{sl}(d, \mathbb{C})$ and consists of traceless complex matrices. Let $E_{ij} \in \mathfrak{sl}(d, \mathbb{C})$ where $i\neq j$ denote the matrix with 1 in the $(i,\,j)$ position and 0 elsewhere. The root system of $\mathfrak{g}$ with respect to $\mathfrak{t}$ is
$\Phi = \{\alpha_{ij}| \quad 1 \leq i \neq j \leq d\}$
where the linear functionals $\alpha_{ij}$ act as
\begin{align}
    \alpha_{ij}: X(\phi) \mapsto \phi_i - \phi_j 
\end{align}
 and the corresponding one-dimensional root spaces are $\mathfrak{g}_{\alpha_{ij}}= \mathbb{C} E_{ij}$. Noting that $\alpha_{ji} = -\alpha_{ij}$ we choose positive roots $\Phi^+=\{\alpha_{ij} | \quad 1 \leq i < j \leq d\}$ and a set of simple roots to be $\Delta=\{\alpha_{i, i+1} | \quad 1 \leq i \leq d-1\}$. We identify the Lie algebra $\mathfrak{t}$ with its dual $\mathfrak{t}^*$ under the inner product obtained from the Killing form
\begin{align}
\label{eq:killing_form}
    (X,Y) &= -2d\tr\left(XY\right) \mathrm{.}
\end{align}
 Under this identification $\alpha_{ij}$ is mapped to a diagonal matrix $X_{\alpha_{ij}}$ from $\mathfrak{t}$ with $\pm i/2d$ appearing as the only two non-zero entries of the $i^\text{th}$ and $j^\text{th}$ positions on the diagonal, respectively. Let $X_{\delta}$ be the element of $\mathfrak{t}$ which is identified with the Weyl vector $\delta$, defined in equation~\eqref{eqn:weyl-vector-delta-defn}. Then 
\begin{align}
    X_{\delta} &= \frac{1}{2}\sum_{i<j} X_{\alpha_{ij}} \mathrm{,} \\
    (X_{\delta})_{kk} &=  i\left(\frac{d+1}{4d} - \frac{k}{2d}\right)\label{eq:delta} \mathrm{,}\\
    \norm{\delta}^2 &= \norm{X_{\delta}}^2\\
    &=2d \sum_{k=1}^d\left(\frac{d+1}{4d} - \frac{k}{2d}\right)^2\\
    &= \frac{d^2 -1}{24}.
\end{align}
The duals of elements of $\Gamma$ may be indexed by length $d-1$ integer vectors $k$
\begin{align}
    X(k) \coloneqq   2\pi i\begin{pmatrix}
        k_1 \\&k_2\\&&\ddots\\&&&k_{d-1}\\&&&&-\sum_{j=1}^{d-1} k_{j}
    \end{pmatrix} \mathrm{.}
\end{align}

To simplify the notation we rename $X(\phi)$ and $X(k)$ as $X_{\phi}$ and $X_k$. Specialising the Poisson form of the heat kernel~\eqref{eqn:urakawa-H} to this parametrisation of the maximal torus of $\mathrm{SU}(d)$, one obtains
\begin{align}
     H_S(\exp(X_{\phi}),\sigma)\coloneqq C(d, \sigma) j(\exp(X_\phi))^{-1} \sum_{k\in\mathbb{Z}^{d-1}} \pi(X_\phi + X_k) \exp\left(-\frac{1}{4\sigma} \norm{X_\phi + X_k}^2\right) \mathrm{,}\label{eqn:heat-kernel-poisson-specialised}
\end{align}
where
\begin{equation}
\label{eq:C_def}
    C(d, \sigma) \coloneqq \frac{c}{\pi(\delta)} (2\pi)^{l+m} i^m \exp\left(\norm{\delta}^2 \sigma\right) (4\pi \sigma)^{-N/2}
\end{equation}
and for convenience we have written everything in terms of elements of the Lie algebra, converting elements of the dual where necessary, e.g. $\pi(X) = \prod_{\alpha\in\Phi^+} (\alpha, \lambda_X) = \prod_{\alpha\in\Phi^+} \alpha(X)$.

In order to obtain the corresponding heat kernel on $\mathrm{PU}(d)$ we proceed as above, and again we average this expression over the normal subgroup $\Gamma$ given by  $d^\text{th}$ roots of unity. We obtain an expression for the heat kernel on $\mathrm{PU}(d)$ in the Poisson form

\begin{align}H_P(\exp(X_{\phi}),\sigma) &\coloneqq \frac{1}{\abs{\Gamma}}\sum_{\gamma \in\Gamma}  H_S(\gamma \exp(X_{\phi}),\sigma) \\
&= 
     \frac{C(d, \sigma)}{\abs{\Gamma}}\sum_{\gamma\in\Gamma} j(\exp(X_\phi))^{-1} \sum_{k\in\mathbb{Z}^{d-1}} \pi(X_\phi + X_k +\log(\gamma)) \exp\left(-\frac{1}{4\sigma} \norm{X_\phi + X_k+\log(\gamma)}^2\right),
\end{align}
where we have used that $j(\gamma e^X) = j(e^X)$ and to match how we parametrised the torus in~\eqref{eqn:maximal-torus-parameterisation} we choose the logarithm to be 
\begin{align}
    \log\left(e^{i\frac{2\pi r}{d}} I \right) = i\frac{2\pi}{d}\begin{pmatrix}
        r\\&\ddots\\&&r\\&&&-r(d-1)
    \end{pmatrix}.
\end{align}

We stress that formally $H_P(\exp(X_{\phi}),\sigma)$, is a function on $\mathrm{SU}(d)$ that is a lift of the heat kernel on $\mathrm{PU}(d)$.

\section{Bounds for exact \texorpdfstring{$t$}{t}-designs}
\label{sec:t-design-eps-net-bound}
In this section, we prove an error bound for a polynomial approximation to the heat kernel. In order to connect to projective $t$-designs, it is necessary for this approximation to be in terms of \emph{balanced polynomials}. We call a function on $\mathrm{PU}(d)$ a balanced polynomial of order $t$ if 
\begin{align}
    f(\mathbf{U}) = \tr\left(\left(U \otimes U^{*}\right)^{\otimes t} A\right),
\end{align}
holds for all $U\in \pi^{-1}(\mathbf{U})$, where $A$ is some fixed matrix and recalling that in our present notation each $ \mathbf{U} \in\mathrm{PU}(d)$ is an equivalence class of elements  $\pi^{-1}(\mathbf{U}) \subset \mathrm{SU}(d)$.

The approximation we seek follows directly from the formula
\begin{align}
    H_P(g,\sigma) = \sum_\lambda d_\lambda \exp(-\sigma k_\lambda) \chi_\lambda(g),\label{eqn:heat-kernel-character-formula-repeated}
\end{align}
where $\sum_i \lambda_i=0$, given by~\eqref{eqn:heat-kernel-character-formula}, upon noticing that each character $\chi_\lambda$, of the projective unitary group is a balanced polynomial of order $\frac{\norm{\lambda}_1}{2}$. This follows since they may be written in terms of Schur functions, see e.g. Ref.~\cite{faraut_2008} for details. Let us denote by $H_P^{(t)}(g,\sigma)$ the restriction of the sum in~\eqref{eqn:heat-kernel-character-formula-repeated} to balanced polynomials of order at most $t$, that is
\begin{align}
\label{eq:trimmed_character}
    H_P^{(t)}(g,\sigma) &= \sum_{\lambda,\, \norm{\lambda}_1 \leq 2t} d_\lambda \exp(-\sigma k_\lambda) \chi_\lambda(g) \mathrm{,}
\end{align}
where $\sum_i \lambda_i=0$.
We refer to the polynomial approximations $H_P^{(t)}$ of the heat kernel $H_P$ as trimmed heat kernels.

We seek to bound the $2$-norm of the difference between the trimmed heat kernel $H_P^{(t)}$ and the full heat kernel $H_P$, where the $2$-norm here is the one induced by the Haar measure
\begin{align}
    \norm{f}_2^2 = \int_{\mathrm{SU}(d)} \abs{f}^2\,d\mu.
\end{align}
Using expressions \eqref{eqn:heat-kernel-character-formula-repeated} and \eqref{eq:trimmed_character} one may bound the trimming error $\norm{H_P(\cdot, \sigma) - H_P^{(t)}(\cdot,\sigma)}_2$ for $t$ large enough (for fixed $\sigma$ and $d$). For a precise statement and proof, see Lemma \ref{lem:trim-heat-kernel} from Appendix \ref{app:trimming}. This allows us to focus on the properties of the full heat kernel $H_P$.

\begin{rem} [Optimality of the trimming procedure]
    The trimming procedure given by~\eqref{eq:trimmed_character} is optimal in the following sense. The trimmed heat kernel $H_P^{(t)}(\cdot, \sigma)$ is the unique function in $\mathcal{H}_t$ closest to the heat kernel $H_P(\cdot, \sigma)$ in $L^2$-norm. Indeed, $H_P^{(t)}(\cdot, \sigma)$ is the orthogonal projection of $H_P(\cdot, \sigma)$ onto a finite-dimensional subspace $\mathcal{H}_t$ of the Hilbert space $L^2(\mathrm{PU}(d))$. Hence, the result follows from Hilbert's projection theorem.
\end{rem}

The next step is to bound the complement of the integral of an absolute value of a heat kernel on $\mathrm{PU}(d)$ over the complement of a small ball $B_{P,\epsilon}$. As explained in Section~\ref{sec:connecting-pun-and-sun}, we reduce this problem to considerations on $\mathrm{SU}(d)$. Recall that an element of $\mathrm{PU}(d)$ consists of an equivalence class of elements of $\mathrm{SU}(d)$ where two matrices are equivalent if they differ by an element of $\Gamma$.
By $B_{\epsilon}(V)$ we denote the closed operator-norm $\epsilon$-ball in $\mathrm{SU}(d)$ centred at $V$

\begin{align}
    B_{\epsilon}(V) &\coloneqq \left\{U\in\mathrm{SU}(d)\,\middle|\,d(U,V) \leq \epsilon\right\}\label{eqn:special-epsilon-ball}.
\end{align}
By $B_{P,\epsilon}(\mathbf{V})$ be denote a closed $\epsilon$-ball centred at $\mathbf{V}$ in metric $d_P(\cdot, \cdot)$
\begin{align}
    B_{P,\epsilon}(\mathbf{V}) &\coloneqq\left\{\mathbf{U}\in\mathrm{PU}(d)\,\middle|\,d_P(\mathbf{U},\mathbf{V}) \leq \epsilon\right\}\label{eqn:projective-epsilon-ball}.
\end{align}
By $\tilde{B}_{P,\epsilon}(V)\subseteq \mathrm{SU}(d)$, where $V\in \pi^{-1}(\mathbf{V})$, we denote the inverse image of $B_{P,\epsilon}(\mathbf{V})$ under the quotient map $\pi: \mathrm{SU}(d) \rightarrow \mathrm{PU}(d)$
\begin{align}
\tilde{B}_{P,\epsilon}(V) &\coloneqq \pi^{-1}({B}_{P,\epsilon}(\mathbf{V}))= \bigcup_{\gamma\in\Gamma} \gamma B_\epsilon(V).\label{eqn:projective-epsilon-ball-in-sun}
\end{align}
If the centre $\mathbf{V}$ is not specified, the ball is centred at the group identity.
By $\mathcal{H}^d_r$ we denote an $\infty$-norm closed ball/hypercube in $\mathbb{R}^d$ of radius $r$
\begin{equation}
  \mathcal{H}^d_r \coloneqq \{v \in \mathbb{R}^d \,|\, ||v||_{\infty} \leq r \}  
\end{equation}
and by $\mathcal{Z}^{d-1}$ we denote the hyperplane in $\mathbb{R}^{d}$ consisting of vectors $y$ with $\sum_{j=1}^d y_j = 0$.

Every element $U \in \mathrm{SU}(d)$ can be written as $U=VDV^{-1}$ for some $V \in \mathrm{SU}(d)$ and $D \in T$. Since the operator norm is unitary invariant, a ball $B_{\epsilon}$ corresponds to a unique ball $B_{T,\epsilon} \subset T$, via  $B_{\epsilon} = \bigcup_{V \in \mathrm{SU}(d)} VB_{T,\epsilon} V^{-1} $, where
$B_{T,\epsilon} = \{D \in T\, |\, d(D,I) \leq \epsilon\}$.

Hence, a ball $B_{\epsilon} \subset \mathrm{SU}(d)$ corresponds to a ball in  $T \subset \mathrm{SU}(d)$ which is an image of $\mathcal{H}^{d-1}_{\tilde{\epsilon}}$ (identified with a subset of $\mathfrak{t}$), under the exponential map $\exp: \mathfrak{t} \rightarrow T$, where
\begin{equation}
    \tilde{\epsilon} = 2 \cdot \mathrm{arcsin}(\epsilon/2) \in [0, \pi] \mathrm{,}
\end{equation}
so $\epsilon \leq \tilde{\epsilon}$.

We first prove a lemma allowing us to remove the summation over $\Gamma$ obtained when we express the $\mathrm{PU}(d)$ heat kernel in terms of that of $\mathrm{SU}(d)$.
\begin{lem}
\label{lem:intP}
    Let $\varphi$ be a non-negative function on $\mathrm{SU}(d)$ Haar-normalised to 1. Fix $\epsilon > 0$ and consider a set $\tilde{B}_{P,\epsilon}$ defined by~\eqref{eqn:projective-epsilon-ball-in-sun}, let its complement be $\tilde{B}_{P,\epsilon}^c$. Then
    \begin{align}
        \int_{\tilde{B}_{P, \epsilon}^c} \frac{1}{|\Gamma|} \sum_{\gamma \in \Gamma}   \varphi (\gamma g) d \mu ( g)  \leq \int_{B_{\epsilon}^c}  \varphi (g) d \mu(g).
    \end{align}
\end{lem}

\begin{proof}
    \begin{align}
  \int_{\tilde{B}_{P, \epsilon}} \frac{1}{|\Gamma|} \sum_{\gamma \in \Gamma}   \varphi (\gamma g) d \mu (g) &=  \frac{1}{|\Gamma|} \sum_{\gamma \in \Gamma} \int_{\cup_{\kappa \in \Gamma} \kappa B_{\epsilon}}  \varphi (\gamma g)  d\mu (g)\\
  &=  \frac{1}{|\Gamma|} \sum_{\gamma \in \Gamma} \int_{\cup_{\kappa \in \gamma\Gamma}  \kappa B_{\epsilon}} \varphi (g)  d\mu (\gamma^{-1}g)\\
   &=  \frac{1}{|\Gamma|} \sum_{\gamma \in \Gamma} \int_{\cup_{\kappa \in \Gamma}  \kappa B_{\epsilon}} \varphi (g)  d\mu (g)\\
  &= \int_{\cup_{\kappa \in \Gamma} \kappa B_{\epsilon}} \varphi (g)  d\mu (g) \\
  &\geq \int_{ B_{\epsilon}} \varphi (g)  d\mu (g) 
\end{align}
hence
\begin{equation}
   1- \int_{\tilde{B}_{P, \epsilon}} \frac{1}{|\Gamma|} \sum_{\gamma \in \Gamma}   \varphi (\gamma g) d \mu (g) \leq  1 - \int_{ B_{\epsilon}} \varphi (g)  d\mu (g) \mathrm{.} 
\end{equation}

The bound in Lemma \ref{lem:intP} may seem crude, however, the more of a mass of $\varphi$ is concentrated in a ball $B_{\epsilon}$ the tighter it becomes. This corresponds e.g. to the heat kernel $H_S(g, \sigma)$ with decreasing $\sigma$ (see also Fig. \ref{fig:summation}).
\end{proof}
\begin{figure}[H]
  \centering
      \includegraphics[width=0.75\textwidth]{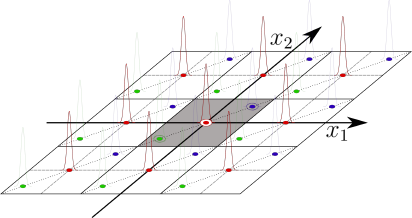}
  \caption{Illustration of the distribution of the components of a heat kernel $H_P$ on $\mathrm{PU}(3)$, obtained via the averaging map applied to a heat kernel $H_S$ in Poisson form. Actual shapes and relative sizes are not depicted. The averaging takes place over $\Gamma$, which consists of three roots of unity, denoted by red, green and blue points in the central square region. The elements of $\Gamma$ act by shifting by the roots of unity along the dotted grey lines, corresponding to a torus. The heat kernel $H_S$ corresponds to the red peaks. Each repeated square region corresponds to the contribution from a different $k$-vector in the Poisson form, which lies on the grey dashed grid. Notice that only the central square region ($k=0$) corresponds to points in a group. However, the tails of the peaks from non-central square regions ($k \neq 0$) overlap with the central square region, contributing to the heat kernel. A ball $\tilde{B}_{P, \epsilon}$ corresponds to a sum of three balls in a central region, denoted by dotted lines. A ball $B_{\epsilon}$ corresponds to the red ball at the origin, and the grey region corresponds to its complement. Lemma \ref{lem:intP} states that the integral of $H_P$ over $\tilde{B}_{P, \epsilon}$ can be upper bounded by the integral of $H_S$ (proportional to the red component) over the grey region. This is outlined by the opacity of the blue and green peaks. Lemma \ref{lem:tail-bound-smaller-than-dominant} shows that this integral can be bounded by bounding the contribution from the central ($k=0$) red peak, which is obtained in Lemma \ref{lem:I0_bound}.}
  \label{fig:summation}
\end{figure}

Applying Lemma \ref{lem:intP} with $\varphi = H_S$ and the Weyl integration formula (see e.g. Ref.~\cite{faraut_2008}) we can bound the integral of $H_P(g, \sigma)$ over $\tilde{B}^c_{P,\epsilon}$ as follows
\begin{align}
  \int_{ \tilde{B}_{P, \epsilon}^c}  H_P(g,\sigma) d \mu(g) &\leq \int_{B^c_{\epsilon}}  H_S(g,\sigma) d \mu(g) \\
  &=  \frac{C(d, \sigma)}{\abs{W}} \sum_{k\in\mathbb{Z}^{d-1}}\int_{\mathcal{H}^{d-1}_{\pi} \setminus \mathcal{H}^{d-1}_{\tilde{\epsilon}}} d\mu(\phi) j(\exp(X_\phi))^* \pi(X_\phi + X_k) \exp\left(-\frac{1}{4\sigma} \norm{X_\phi + X_k}^2\right),
\end{align}
where $d\mu(\phi)=\frac{d \phi_1 d \phi_2 \ldots d \phi_{d-1}}{(2\pi)^{d-1}}$ stems from the Haar measure on $T$, $W$ is the Weyl group
and we have cancelled the $j^{-1}$ with part of the $\abs{j}^2$ term in the Weyl integration formula. 

Using the triangle inequality, we obtain
\begin{align}
\label{eq:triangle}
 \int_{ \tilde{B}_{P, \epsilon}^c}  |H_P(g,\sigma)| d \mu(g)  \leq  \frac{C(d, \sigma)}{\abs{W}} \sum_{k\in\mathbb{Z}^{d-1}}\int_{\mathcal{H}^{d-1}_{\pi} \setminus \mathcal{H}^{d-1}_{\tilde{\epsilon}}} d\mu(\phi) |j(\exp(X_\phi)) \pi(X_\phi + X_k)| \exp\left(-\frac{1}{4\sigma} \norm{X_\phi + X_k}^2\right).
\end{align}

We seek to express the right-hand side of~\eqref{eq:triangle} in terms of the dominant term $\mathcal{I}_0$ ($k=0$) and some smaller correction $\mathcal{R}$ which we will bound in terms of $\mathcal{I}_0$
\begin{align}
\label{eq:triangle2}
  \int_{ \tilde{B}_{P, \epsilon}^c} | H_P(g,\sigma)| d \mu(g) \leq  \mathcal{I}_0 + \mathcal{R},
\end{align}
where
\begin{equation}
    \mathcal{I}_0 \coloneqq \frac{C(d, \sigma)}{\abs{W}} \int_{\mathcal{H}^{d-1}_{\pi} \setminus \mathcal{H}^{d-1}_{\tilde{\epsilon}}} d\mu(\phi) |j(\exp(X_\phi)) \pi(X_\phi)| \exp\left(-\frac{1}{4\sigma} \norm{X_\phi }^2\right),
    \label{eq:I0}
\end{equation}
and 
\begin{equation}
    \mathcal{R} \coloneqq \frac{C(d, \sigma)}{\abs{W}} \sum_{k\in\mathbb{Z}^{d-1}\setminus \{0\}}\int_{\mathcal{H}^{d-1}_{\pi} \setminus \mathcal{H}^{d-1}_{\tilde{\epsilon}}} d\mu(\phi) |j(\exp(X_\phi)) \pi(X_\phi + X_k)| \exp\left(-\frac{1}{4\sigma} \norm{X_\phi + X_k}^2\right).
      \label{eq:R}
\end{equation}

We provide the bounds on $\mathcal{I}_0$ and $\mathcal{R}$ via Lemma \ref{lem:I0_bound} proved in Appendix \ref{app:I0} and Lemma \ref{lem:tail-bound-smaller-than-dominant} in Appendix \ref{app:R} respectively. Our bounds apply for $\sigma$ small enough (for fixed $\epsilon$ and $d$).

By combining the trimming error bound with the vanishing properties of the full heat kernel $H_P$ (Lemmas~\ref{lem:trim-heat-kernel},~\ref{lem:I0_bound} and~\ref{lem:tail-bound-smaller-than-dominant}) we obtain a bound for the vanishing of the absolute value of the trimmed heat kernel  $H_P^{(t)}$, stated in Lemma  \ref{lem:combined-bounds}.

\begin{lem}\label{lem:combined-bounds}
Provided that

\begin{align}
    2t \geq \frac{d^2}{\sqrt{\sigma}}\sqrt{2\log\left(\frac{d^4}{\sigma}\right)}\label{eqn:t-bound}
\end{align}
and 
\begin{align}
    \sigma \leq \frac{\epsilon^2}{ 32 d \mathrm{log}(d)}\label{eqn:sigma-bound}
\end{align}
for any
\begin{align}
   \eta \geq \frac{1}{\prod_{k=1}^{d}k!}
\end{align}
we have
\begin{align}
    \int_{\tilde{B}_{P, \epsilon}^c} \abs{H_P^{(t)}(g, \sigma)} d\mu(g) &\leq  2^{\frac{d}{2}}\exp\left(-\sigma\frac{t^2}{d^2} - \frac{1}{2}\sigma t\right) +  \frac{1+\eta}{2}\exp\left(-\frac{d}{16 \sigma} \epsilon^2 + \frac{d^2-1}{24} \sigma \right) \mathrm{.}
\end{align}
\end{lem}
\begin{proof}
   
    \begin{align}
      \int_{\tilde{B}_{P, \epsilon}^c} \abs{H_P^{(t)}(g, \sigma)} d\mu(g) &\leq \int_{\tilde{B}_{P, \epsilon}^c} |H_P^{(t)}(g, \sigma) - H_P(g, \sigma)| d \mu(g) + \int_{\tilde{B}_{P, \epsilon}^c} | H_P(g, \sigma)| d \mu(g)\\
        &= \int_{\mathrm{SU}(d)} |H_P^{(t)}(g, \sigma) - H_P(g, \sigma)| \chi_{\tilde{B}_{P, \epsilon}^c}(g)d \mu(g) + \int_{\tilde{B}_{P, \epsilon}^c} | H_P(g, \sigma)| d \mu(g), \label{eq:ho_splitted}
    \end{align}
    where $\chi_X$ denotes the indicator function of the set $X$. Applying H{\"o}lder's inequality to the first term of~\eqref{eq:ho_splitted} gives
    \begin{align}
         \int_{\tilde{B}_{P, \epsilon}^c} \abs{H_P^{(t)}(g, \sigma)} d\mu(g) &\leq \norm{H_P^{(t)}(g, \sigma) - H_P(g, \sigma)}_2  + \int_{\tilde{B}_{P, \epsilon}^c} | H_P(g, \sigma)| d \mu(g).
    \end{align}
    Finally, substituting the bounds from Lemmas~\ref{lem:trim-heat-kernel},~\ref{lem:I0_bound} and~\ref{lem:tail-bound-smaller-than-dominant} with $\gamma=1/2$ and applying $\epsilon \leq \tilde{\epsilon}$ gives the result. The condition~\eqref{eqn:sigma-bound} is the result of multiplying the bounds on $\sigma$ we require for each lemma; one could obtain a slightly improved, but more complicated, bound by taking the minimum rather than the product.
\end{proof}

We are now able to prove our first theorem
\begin{thm}
\label{th:theorem1}
    Let $\nu$ be an exact $t$-design in $\mathrm{PU}(d)$, $d \geq 2$, then $\operatorname{supp}(\nu{)}$ is an $\epsilon$-net provided
    \begin{align}
        t \geq 32\frac{d^{\frac{5}{2}}}{\epsilon}\log(d)\log\left(\frac{4}{a_v \epsilon}\right)\mathrm{,}\label{eqn:theorem-1-t-bound}
    \end{align}
    where $C =  9 \pi$.
\end{thm}
\begin{proof}
    We proceed via a proof by contradiction. Assume $\operatorname{supp}(\nu)$ is not an $\epsilon$-net, then according to Ref.~\cite{9614165}, Lemmas~1 and~2, we know there exists a $\mathbf{V_0}\in\mathrm{PU}(d)$ such that for any $\kappa\leq \epsilon$
    \begin{equation}
     \mathrm{Vol}(B_{P,\kappa}(\mathbf{V_0}))  \leq \max_{V \in  \tilde{B}^c_{P,\epsilon} } \int_{\tilde{B}_{P,\kappa}(V)} H_P^{(t)}(g, \sigma)d\mu(g).
\end{equation}

Note that $B_{P,\kappa}(\mathbf{V}) \subset B^c_{P,\epsilon-\kappa}$, hence also $\tilde{B}_{P,\kappa}(V) \subset \tilde{B}^c_{P,\epsilon-\kappa}$, so
    \begin{align}
        \max_{V \in  \tilde{B}^c_{P,\epsilon}} \int_{B_{\kappa}(V)} H_P^{(t)}(g, \sigma) d\mu(g) &\leq  \max_{V \in  \tilde{B}^c_{P,\epsilon}} \int_{\tilde{B}_{P,\kappa}(V)}\abs{ H_P^{(t)}(g, \sigma) }d\mu(g) \\&\leq\int_{\tilde{B}^c_{P,\epsilon-\kappa}} \abs{H_P^{(t)}(g, \sigma)} d\mu(g) \mathrm{.}
\end{align}

The Haar ($\mu_P$) volume of $\kappa$-ball (described in metric $d_P(\cdot, \cdot)$ ) in $\mathrm{PU}(d)$ can be bounded from below as follows:
\begin{equation}
    \mathrm{Vol}(B_{P,\kappa}) \geq \left( a_v \kappa \right)^{d^2-1},
    \label{eq:lbv}
\end{equation}
where $a_v = \frac{1}{9 \pi}$ (see Ref.~\cite{9614165}). Such a volume does not depend on the centre of the ball, due to the translation-invariance of the Haar measure and the metric $d_P(\cdot, \cdot)$. We take $\kappa=\frac{\epsilon}{2}$ and therefore have a contradiction if

\begin{equation}
   \int_{\tilde{B}^c_{P,\epsilon/2}} \abs{H^{(t)}(g, \sigma)} d\mu(g) < \left(\frac{1}{2} a_v \epsilon \right)^{d^2-1} \mathrm{.}
\end{equation}

Hence, under the assumptions of Lemma~\ref{lem:combined-bounds} (with $\epsilon/2$ instead of $\epsilon$), with the choice of $\eta=1$, in order to get a contradiction, we can demand for example
\begin{align}
\exp\left(-\frac{d}{64 \sigma} \epsilon^2 + \frac{d^2-1}{24} \sigma \right)&<\frac{1}{2} \left(\frac{1}{2} a_v \epsilon \right)^{d^2-1}\label{eqn:contradiction-ineq-1}\\
   2^{\frac{d}{2}}\exp\left(-\sigma\frac{t^2}{d^2} - \frac{1}{2}\sigma t\right) &< \frac{1}{2}\left(\frac{1}{2}  a_v \epsilon \right)^{d^2-1}\label{eqn:contradiction-ineq-2}.
\end{align}
The inequality~\eqref{eqn:contradiction-ineq-1} constrains $\sigma$ as a function of $\epsilon$ and $d$ and is satisfied whenever
\begin{align}
\sigma < \sigma_{*} = \frac{ \epsilon^2 }{128 d\log(d)\log\left(\frac{2}{a_v \epsilon}\right)},
\end{align}
which may be seen by taking logarithms of both sides of~\eqref{eqn:contradiction-ineq-1} and bounding the term containing $(d^2-1)\frac{\sigma}{24}$ using assumption~\eqref{eqn:sigma-bound}. We can now bound the sufficient $t$, assuming $\sigma = \sigma^*$, using~\eqref{eqn:contradiction-ineq-2}.  Simply taking logarithms of~\eqref{eqn:contradiction-ineq-2} and substituting in $\sigma=\sigma^*$ we obtain
\begin{align}
    t^2 \geq 128\frac{d^5}{\epsilon^2}\log(d)\log^2\left(\frac{4}{a_v \epsilon}\right) \mathrm{,}
\end{align}
however we additionally need $t$ to satisfy the assumption of Lemma \ref{lem:combined-bounds}, so we obtain a final scaling
\begin{align}
    t^2 \geq 1024\frac{d^5}{\epsilon^2}\log^2(d)\log^2\left(\frac{4}{a_v \epsilon}\right) \mathrm{,}
\end{align}
\end{proof}

\section{Bounds for \texorpdfstring{$\delta$}{delta}-approximate \texorpdfstring{$t$}{t}-designs}
\label{sec:delta-t-design-eps-net-bound}

In order to derive the version of Theorem \ref{th:theorem1} for $\delta$-approximate $t$-designs, we bound the $L^2$-norm of the heat kernel.

Since we want to apply the results from the previous sections, we use the Poisson form of the heat kernel, which allows us to group the terms as follows

\begin{align}
\label{eq:l2_split}
     \norm{ H_S(\cdot,\sigma)}_2^2 &=  \frac{C(d, \sigma)^2}{|W|} \int \sum_{k\in\mathbb{Z}^{d-1}} \sum_{l\in\mathbb{Z}^{d-1}} \pi(X_\phi + X_k ) \pi^*(X_\phi + X_l ) e^{-\frac{1}{4\sigma} \left( \norm{X_\phi + X_k)}^2+\norm{X_\phi + X_l)}^2\right)} d\mu(\phi)  \\
     &\leq  \frac{C(d, \sigma)^2}{|W|} \int \sum_{k\in\mathbb{Z}^{d-1}} \sum_{l\in\mathbb{Z}^{d-1}} \abs{ \pi(X_\phi + X_k ) \pi^*(X_\phi + X_l )} e^{-\frac{1}{4\sigma} \left( \norm{X_\phi + X_k)}^2+\norm{X_\phi + X_l)}^2\right)} d\mu(\phi)  \\
     &= \mathcal{I}_{0,0}^2 + \mathcal{R}_{*,0}^2 + \mathcal{R}_{0,*}^2 + \mathcal{R}_{*,*}^2 \mathrm{,}
\end{align}
 where $\mathcal{I}_{0,0}^2$ is the $k=0$ and $l=0$ term, $\mathcal{R}_{*,0}^2$ is the sum of the terms with $k \neq 0$ and $l=0$, $\mathcal{R}_{0,*}^2$ is the sum of the terms with $k = 0$ and $l \neq 0$ and $\mathcal{R}_{*,*}^2$ is the sum of the terms with $k \neq 0$ and $l \neq 0$. We bound the contributions from $\mathcal{I}_{0,0}^2$, $\mathcal{R}_{*,0}^2$ and $\mathcal{R}_{*,*}^2$ separately in Appendix \ref{app:L2}. The joint bound for $ \norm{ H_P^{(t)}(\cdot,\sigma)}_2$ is provided in Lemma \ref{lem:L2_bound} from Appendix \ref{app:L2}.
 
We now have all the prerequisites to prove our main theorem, which is a generalisation of Theorem \ref{th:theorem1} to $\delta$ approximate $t$-designs.
\begin{thm}
\label{th:theorem2}
    Let $\nu$ be a $\delta$-approximate $t$-design in $\mathrm{PU}(d)$, $d \geq 2$, with
\begin{equation}
      \delta \leq \left(\frac{1}{4C \mathrm{log}^{1/4}\left(\frac{2 C}{  \epsilon}\right) \mathrm{log}^{1/4}(d) }\frac{\epsilon}{d^{1/2}}\right)^{d^2-1}
\end{equation}
where
\begin{equation}
    C = 9\pi \mathrm{,}
\end{equation}
    then $\operatorname{supp}(\nu)$ is an $\epsilon$-net provided
    \begin{align}
         t \geq 32\frac{d^{\frac{5}{2}}}{\epsilon}\log(d)\log\left(\frac{4}{a_v \epsilon}\right)\mathrm{.}
    \end{align}
\end{thm}

\begin{proof}
    We proceed as in the proof of Theorem \ref{th:theorem1}. Assume $\operatorname{supp}(\nu)$ is not an $\epsilon$-net, then according to Ref.~\cite{9614165}, Lemma~2 and 3, we know there exists a $\mathbf{V_0}\in\mathrm{PU}(d)$ such that for any $\kappa\leq \epsilon$
    \begin{align}
     \mathrm{Vol}(B_{P,\kappa}(\mathbf{V_0})) &\leq \int_{\tilde{B}^c_{P,\epsilon-\kappa}} \abs{H_P^{(t)}(g, \sigma)} d\mu(g) + \delta \sqrt{\mathrm{Vol}(B_{P,\kappa}(\mathbf{V_0}))} ||H_P^{(t)}(\cdot, \sigma)||_2 \\
     &\leq  2^{\frac{d}{2}}\exp\left(-\sigma\frac{t^2}{d^2} - \frac{1}{2}\sigma t\right) + \frac{1+ \left( 1+ \delta \sqrt{\mathrm{Vol}(B_{P,\epsilon/2}}) \frac{d\sqrt{d!}}{2^{m-1}}\right) \frac{1}{\prod_k k!} }{2}  \exp\left(-\frac{d}{64  \sigma} \epsilon^2 + \frac{d^2-1}{24} \sigma \right)\\ &+ \delta \sqrt{\mathrm{Vol}(B_{P,\epsilon/2}}) d \mathcal{I}_{0,0} \\
     &\leq  2^{\frac{d}{2}}\exp\left(-\sigma\frac{t^2}{d^2} - \frac{1}{2}\sigma t\right) + \exp\left(-\frac{d}{64  \sigma} \epsilon^2 + \frac{d^2-1}{24} \sigma \right) + \delta \sqrt{\mathrm{Vol}(B_{P,\epsilon/2}}) d \mathcal{I}_{0,0}  
\end{align}
where we put $\kappa=\epsilon/2$ and used Lemmas \ref{lem:combined-bounds} and \ref{lem:L2_bound}  with $\eta =\frac{1}{\prod_{k=1}^d k!} $. Moreover, we assumed $\delta$ is not too large, so that
\begin{equation}
    \left( 1+ \delta \sqrt{\mathrm{Vol}(B_{P,\epsilon/2}}) \frac{d\sqrt{d!} }{2^{m-1}}\right) \frac{1}{\prod_{k=1}^d k!}\leq 1 \mathrm{,}
\end{equation}
e.g.
\begin{equation}
    \label{eq:delta_l2_cond}
    \delta \leq\frac{1}{2\sqrt{2}} \leq \frac{2^{m-1}}{d \sqrt{d!} \sqrt{\mathrm{Vol}(B_{P,\epsilon/2})}}  \mathrm{.}
\end{equation}
We take $\kappa=\frac{\epsilon}{2}$ and therefore have a contradiction if, under the assumptions of Lemma~\ref{lem:combined-bounds} we have three inequalities
\begin{align}
\exp\left(-\frac{d}{64 \sigma} \epsilon^2 + \frac{d^2-1}{24} \sigma \right)&<\frac{1}{2} \left(\frac{1}{2} a_v \epsilon \right)^{d^2-1}\label{eqn:contradiction-delta-ineq-1}\\
     2^{\frac{d}{2}}\exp\left(-\sigma\frac{t^2}{d^2} - \frac{1}{2}\sigma t\right) &< \frac{1}{4}\left(\frac{1}{2}  a_v \epsilon \right)^{d^2-1}\label{eqn:contradiction-delta-ineq-2}\\
  \delta d  \mathcal{I}_{0,0} &< \frac{1}{4}\left(\frac{1}{2}  a_v \epsilon \right)^{\frac{d^2-1}{2}}\label{eqn:contradiction-delta-ineq-3}.
\end{align}
Inequality~\eqref{eqn:contradiction-delta-ineq-1} is the same same as~\eqref{eqn:contradiction-ineq-1} from the proof of Theorem \ref{th:theorem1}, hence it is satisfied for the same $\sigma = \sigma^{*}$. Inequality~\eqref{eqn:contradiction-delta-ineq-2} differs from~\eqref{eqn:contradiction-ineq-2} by the factor of $1/4$ instead of $1/2$, one may check that this inequality is still satisfied as long as $t$ satisfies the bound in equation~\eqref{eqn:theorem-1-t-bound}.

It remains to ensure that~\eqref{eqn:contradiction-delta-ineq-3} is satisfied.
Using Lemmas \ref{lem:I^2}, \ref{lem:C_value} and \ref{lem:Theta}, then taking the logarithms of both sides of~\eqref{eqn:contradiction-delta-ineq-3} and bounding the terms that do not depend on $\epsilon$ or $\sigma$ by the $\Theta(d^2 \mathrm{log}(d))$ term, we obtain

\begin{equation}
    \mathrm{log}\left(\frac{1}{\delta}\right) \geq \left(\frac{d^2-1}{4} \right) \mathrm{log}\left(\frac{1}{\sigma}\right) +  \frac{d^2-1}{24} \sigma + \left(\frac{3d^2}{16} + 4\right) \mathrm{log}(d) + \frac{d^2-1}{2}\mathrm{log}\left(\frac{2}{a_v \epsilon}\right) \mathrm{.}
\end{equation}
Plugging $\sigma = \sigma^{*}$ leads to

\begin{equation}
      \delta \leq \left(\frac{ a_v}{ 2^{9/2}} \right)^{\frac{d^2-1}{2}} \left(\frac{\epsilon}{\mathrm{log}^{1/4}\left(\frac{2}{a_v \epsilon}\right) \mathrm{log}^{1/4}(d)}\right)^{d^2-1} \frac{\mathrm{exp}\left(-\frac{(d^2-1)\epsilon^2}{3072 d  \mathrm{log}(d) \mathrm{log}\left(\frac{2}{a_v \epsilon}\right)}\right)}{d^{\frac{7}{16}d^2+\frac{15}{4}}}\label{eqn:delta-bound-with-exponential} \mathrm{.}
\end{equation}
One may check that~\eqref{eqn:delta-bound-with-exponential} is stronger than~\eqref{eq:delta_l2_cond}. Moreover, since $\epsilon \leq 2$, we can lower bound the exponential term as

\begin{equation}
    \mathrm{exp}\left(-\frac{(d^2-1)\epsilon^2}{3072 d  \mathrm{log}(d) \mathrm{log}\left(\frac{2}{a_v \epsilon}\right)}\right) \geq \mathrm{exp}\left(-\frac{d}{768 \mathrm{log}(d) \mathrm{log}\left(\frac{1}{a_v }\right)}\right) = d^{-\frac{d}{768 \mathrm{log}^2(d) \mathrm{log}\left(\frac{1}{a_v } \right)}} \mathrm{.}
\end{equation}
Hence,

\begin{equation}
      \delta \leq \left(\frac{ a_v}{ 2^{9/2}} \right)^{\frac{d^2-1}{2}}\left(\frac{\epsilon}{\mathrm{log}^{1/4}\left(\frac{2}{a_v \epsilon}\right) \mathrm{log}^{1/4}(d) d^{1/2}}\right)^{d^2-1} \kappa(d) \label{eqn:delta-bound-with-kappa} \mathrm{,}
\end{equation}
where
\begin{equation}
    \kappa(d) \coloneqq d^{\frac{d^2}{16}-\frac{17}{4}-\frac{d}{768 \mathrm{log}^2(d) \mathrm{log}\left(\frac{1}{a_v} \right)}} \mathrm{.}
\end{equation}
\end{proof}
We now observe that the function $\kappa(d)^{\frac{1}{d^2-1}}$ is increasing for $d\geq 2$, which may be demonstrated by computing the derivative. We may, therefore, lower bound it by bounding the value at $d=2$. For example,
\begin{align}
    \kappa(d)^{\frac{1}{d^2-1}} \geq 2^{-\frac{3}{2}}. 
\end{align}
Combining this bound with the above reasoning, we obtain
\begin{align}
    \delta \leq \left(\frac{1}{16\sqrt{9\pi} \mathrm{log}^{1/4}\left(\frac{2}{a_v \epsilon}\right) \mathrm{log}^{1/4}(d) }\frac{\epsilon}{d^{1/2}}\right)^{d^2-1} \mathrm{,}
\end{align}
which may easily be seen to imply the bound shown in the Theorem statement.
\begin{rem}
\label{rem:main_result}
The bound on $\delta$ provided in Theorem~\ref{th:theorem2} is significantly looser than e.g.~\eqref{eqn:delta-bound-with-exponential} or~\eqref{eqn:delta-bound-with-kappa}. In the provided form Theorem~\ref{th:theorem2} is appropriate for comparison with the results of Ref.~\cite{9614165}, however in applications we expect one of the more precise bounds to be more appropriate.
\end{rem}

\FloatBarrier
\section{Trimmed heat kernel as a polynomial approximation of Dirac delta}
\label{sec:approx-delta}

In this section, we summarise various properties of our construction of the polynomial approximation of the Dirac delta. These properties are either obvious or were addressed in previous sections. The only missing property was the behaviour of the $L^1$-norm of the trimmed heat kernel, which also bounds its negativity. 

\begin{thm}
 \label{th:theorem3}
     The trimmed heat kernel $H_P^{(t)}(\cdot, \sigma)$ for $\mathbf{U}(d)$ with $d \geq 2$ has the following properties:
     \begin{enumerate}
       \item $H_P^{(t)}(\cdot, \sigma) \in \mathcal{H}_t$.
        \item $H_P^{(t)}(\cdot, \sigma)$ is Haar-normalised to 1 and also approximately non-negative for $t$ large enough (see point \ref{p:L1}).
        
        \item Controllable vanishing outside the ball of radius $\epsilon$ as $\sigma \to 0$ 
        \begin{align}
    \int_{\tilde{B}_{P, \epsilon}^c} \abs{H_P^{(t)}(g, \sigma)} d\mu(g) &\leq  2^{\frac{d}{2}}\exp\left(-\sigma\frac{t^2}{d^2} - \frac{1}{2}\sigma t\right) +  \frac{1+\eta}{2}\exp\left(-\frac{d}{16 \sigma} \epsilon^2 + \frac{d^2-1}{24} \sigma \right)
\end{align}
        for \begin{align}
    t \geq t_{*} \coloneqq  \frac{d^2}{2 \sqrt{\sigma}}\sqrt{2\log\left(\frac{d^4}{\sigma}\right)}
\end{align}
and 
\begin{align}
    \sigma \leq \frac{\epsilon^2}{ 32 d \mathrm{log}(d)}
\end{align} with $\eta \geq \frac{1}{\prod_{k=1}^{d}k!} $. Hence, for any $\epsilon > 0$, 
        \begin{equation}
           \lim_{\sigma \to 0} \int_{ \tilde{B}_{P, \epsilon}^c}  |H_P^{(t_*)}(g,\sigma)| d \mu(g) =0 \mathrm{.}
        \end{equation}
        \item Controllable blow-up of the $L^2$-norm
        \begin{equation}
            ||H_P^{(t)}(\cdot, \sigma)||_2 \leq c \left(\frac{d}{\sigma}\right)^{\frac{d^2-1}{4}} 
        \end{equation}
         for $\sigma \leq \frac{1}{d \mathrm{log}(d)}$ and some positive group constant $c$. One can take $c = 8$ for $d\geq 2$ and  $c=1$ for $d\geq 12$.
        \item \label{p:L1}  Bounded $L^1$-norm
        \begin{equation}
         ||H_P^{(t)}(\cdot, \sigma)||_1 \leq 1 +  2^{\frac{d}{2}}\exp\left(-\sigma\frac{t^2}{d^2} - \frac{1}{2}\sigma t\right)
        \end{equation}
         for $t \geq t_{*} \mathrm{.}$ Hence, for fixed $d$
        \begin{equation}
             \lim_{\sigma \to 0} ||H_P^{(t_*)}(\cdot, \sigma)||_1 =1 \mathrm{.}
        \end{equation}
     \end{enumerate}
 \end{thm}

\begin{proof}
    Point 1 follows from the construction detailed above. Point 2 is a consequence of the orthogonality of characters and point 5. Point 3 is the Lemma \ref{lem:combined-bounds}. Point 4 is Corollary \ref{lem:L2_bound2} from Appendix \ref{app:L2}. Point 5 can be proved using the triangle inequality and H{\"o}lder's inequality. Indeed, since $H_P(\cdot, \sigma)$ is normalised to 1 and non-negative, using Lemma \ref{lem:trim-heat-kernel} we can write
    \begin{align}
        ||H_P^{(t)}(\cdot, \sigma)||_1 &\leq ||H_P(\cdot, \sigma)||_1 +  ||(H_P^{(t)}(\cdot, \sigma)-H_P(\cdot, \sigma)) \chi_{\mathrm{SU}(d)}||_1\\
        &\leq 1 + ||H_P^{(t)}(\cdot, \sigma)-H_P(\cdot, \sigma)||_2\\
        & \leq 1 +  2^{\frac{d}{2}}\exp\left(-\sigma\frac{t^2}{d^2} - \frac{1}{2}\sigma t\right)
    \end{align}
    for \begin{align}
    t \geq  \frac{d^2}{2\sqrt{\sigma}}\sqrt{2\log\left(\frac{d^4}{\sigma}\right)} \mathrm{.}
    \end{align}
\end{proof}


\begin{rem}
One can also write down similar properties of the (full) heat kernel $H_P(\cdot, 
\sigma)$, which follow from the proofs of the same Lemmas as Theorem \ref{th:theorem3}. In this case, point 1 is not true for any $t$. The normalisation from point 2 is true, but $H_P(\cdot, 
\sigma)$ is non-negative so $||H_P(\cdot, \sigma)||_1=1$. The bound from point 4 is valid. The bound from point 3 simplifies to

\begin{align}
    \int_{\tilde{B}_{P, \epsilon}^c} \abs{H_P(g, \sigma)} d\mu(g) &\leq  \frac{1+\eta}{2}\exp\left(-\frac{d}{16 \sigma} \epsilon^2 + \frac{d^2-1}{24} \sigma \right)
\end{align}
for
\begin{align}
    \sigma \leq \frac{\epsilon^2}{ 32 d \mathrm{log}(d)}
\end{align} and $\eta \geq \frac{1}{\prod_{k=1}^{d}k!} $. Hence, for any $\epsilon > 0$
        \begin{equation}
           \lim_{\sigma \to 0} \int_{ \tilde{B}_{P, \epsilon}^c}  |H_P(g,\sigma)| d \mu(g) =0 \mathrm{.}
        \end{equation}
\end{rem}

\begin{rem}
Formally, Theorem \ref{th:theorem3} shows in particular that for fixed $d\geq 2$, the family of $L^1$-integrable functions $\{k_{\lambda}\}_{\lambda > 0}$, where $k_{\lambda}(g) \coloneqq H_P^{t_*(\lambda)}(g, 1/\lambda)$ and $t_*(\lambda) \coloneqq \frac{d^2}{2}\sqrt{2 \lambda \log\left( d^4 \lambda \right)}$, is an approximate identity on $\mathrm{PU}(d)$.
\end{rem}

\section{Summary and future work}
\label{sec:summary}

We have improved on the state-of-the-art for constructing $\epsilon$-nets in the group of unitary channels from unitary $t$-designs, that of Ref.~\cite{9614165}. Our method involves the construction of a polynomial approximation to a Dirac delta on the space of quantum unitary channels $\mathrm{PU}(d)$ stemming from the natural object - a heat kernel on $\mathrm{SU}(d)$. 

In the case of exact $t$-designs we obtain results very close to those of Ref.~\cite{9614165}, but in the more practically relevant approximate case, our results significantly improve on the state-of-the-art, showing that $\delta$-approximate $t$-designs form $\epsilon$ nets for much larger values of $\delta$ than was previously known.

While our scaling of $\delta$ with $d$ and $\epsilon$ for approximate $t$ designs substantially improves on prior work, it is not obvious that it is optimal. We leave for future work the task of either improving this scaling even further or proving that no improvements are possible. 

The construction of Ref.~\cite{9614165} is used in Ref.~\cite{PhysRevX.14.041068} to prove saturation and recurrence results for the complexity of random quantum circuits without the assumptions on the gap of the universal set $\mathcal{S}$. In future work, it may be possible to apply our construction to improve the known results in this setting. It would also be interesting to understand the unknown constants $A$ and $B$ appearing in~\eqref{eqn:application2-with-group-constants}, obtaining an inverse-free non-constructive Solovay-Kitaev like theorem from our results. This amounts to deriving explicit constants of the polylogarithmic spectral gap decay (see Theorem 6 from Ref.~\cite{varjú2015randomwalkscompactgroups}), especially $r_0$. We expect that the trimmed heat kernel construction may be applied to obtain explicit bounds on such decay.

\section*{Acknowledgments}
This research was funded by the National Science Centre, Poland under the grant OPUS: UMO2020/37/B/ST2/02478. AS would like to thank Luz Roncal for inspiring discussions regarding heat kernels on compact Lie groups. OS would like to thank Michał Oszmaniec for the discussion about our results. 
\appendix

\section{Connecting the heat kernels on the projective and special unitary groups}
\label{app:heat_kenrel_on_su_and_pu}

In this Appendix, we prove Lemma \ref{lem:averaging-irreps-is-nice}, which can be used to show how one can obtain a heat kernel on $\mathrm{PU}(d)$ using the heat kernel on $\mathrm{SU}(d)$ via averaging.

\begin{lem}\label{lem:averaging-irreps-is-nice}
    Let $K$ be a simply-connected compact Lie group and $\Gamma$ a finite normal subgroup so that $K/\Gamma$ is a compact Lie group. Let $\rho$ be a finite-dimensional unitary irrep of $K$. Then the function
\begin{align}
    \tilde{\rho}:g\mapsto \frac{1}{\abs{\Gamma}}\sum_{\gamma\in \Gamma} \rho(\gamma g),
\end{align}
is either 
\begin{enumerate}
    \item identical to $\rho$ if $\rho$ is constant on $\Gamma$-cosets in $K$ or
    \item identically zero.
\end{enumerate}
\end{lem}
\begin{proof}
We observe
\begin{align}
    \tilde{\rho}(g)\tilde{\rho}(g^\prime) &= \tilde{\rho}(g g^\prime), \label{eqn:rho-tilde-respects-multiplication}
\end{align}
so $\tilde{\rho}$ is a group homomorphism exactly if it maps the identity in $K$ to the identity operator. Since $\tilde{\rho}(e)$ is self-adjoint and is easily seen to be idempotent, it is an orthogonal projector. Since $\tilde{\rho}(e)$ commutes with every operator in the image of $K$ under $\rho$ and $\rho$ is irreducible by Schur's lemma, $\rho(e) \propto I$ and is therefore equal to either $I$ or $0$.

In the case that $\tilde{\rho}(e)$ is the identity operator compute
\begin{align}
    \rho(g^{-1}) \tilde{\rho}(g) &= \frac{1}{\abs{\Gamma}}\sum_{\gamma\in \Gamma} \rho(g^{-1}) \rho(\gamma g)\\
    &= \frac{1}{\abs{\Gamma}}\sum_{\gamma\in \Gamma} \rho(g^{-1} \gamma g)\\
    &= \frac{1}{\abs{\Gamma}}\sum_{\gamma^\prime \in \Gamma} \rho(g^{-1} g \gamma^\prime g^{-1} g)\\
    &= \tilde{\rho}(e) = I,
\end{align}
so that for all $g\in K$ we have $\rho(g) = \tilde{\rho}(g)$ implying that the $\rho$ is constant on $\Gamma$-cosets.
\end{proof}

\begin{proof}[Proof of~\eqref{eqn:rho-tilde-respects-multiplication}]
    \begin{align}
    \tilde{\rho}(g)\tilde{\rho}(g^\prime) &= \frac{1}{\abs{H}^2}\sum_{h,h^\prime\in H} \rho(hg)\rho(h^\prime g^\prime)\\
    &= \frac{1}{\abs{H}^2}\sum_{h,h^\prime\in H} \rho(hgh^\prime g^\prime )\\
    &= \frac{1}{\abs{H}^2}\sum_{h\in H} \sum_{h^{\prime\prime}\in H} \rho(hg g^{-1} h^{\prime\prime} gg^\prime)\\
    &= \frac{1}{\abs{H}^2}\sum_{h\in H} \sum_{h^{\prime\prime}\in H} \rho(h h^{\prime\prime} gg^\prime)\\
    &= \frac{1}{\abs{H}}\sum_{h\in H}  \rho(h gg^\prime) = \tilde{\rho}(g g^\prime)
    \end{align}
\end{proof}
 
\section{Bounding the polynomial approximation of the heat kernel}
\label{app:trimming}
In this Appendix, we prove Lemma \ref{lem:trim-heat-kernel}, which quantifies the $L^2$-norm difference between the heat kernel $H_P$ and the trimmed heat kernel $H_P^{(t)}$.

\begin{lem}[``Trimming'' the heat kernel]\label{lem:trim-heat-kernel}
    The trimmed heat kernel $H_P^{(t)}$ satisfies 
    \begin{align}
    \norm{H_P(\cdot, \sigma) - H_P^{(t)}(\cdot,\sigma)}_2 \leq 2^{\frac{d}{2}}\exp\left(-2\sigma(1-\gamma)\frac{t^2}{d^2} - \frac{1}{2}\sigma t\right),\label{eqn:bound-on-trimmed-kernel-approx}
    \end{align}
    for any $0 < \gamma < 1$, provided that 
    \begin{align}
        2t \geq \frac{d^2}{\sqrt{\gamma\sigma}}\sqrt{\log\left(\frac{d^4}{2\gamma\sigma}\right)}.
    \end{align}
\end{lem}
\begin{proof}
     In terms of $L^2$-norm, as a function of $x$ the approximation error may be computed using
\begin{align}
    \norm{H_P(\cdot, \sigma) - H_P^{(t)}(\cdot,\sigma)}_2^2 &= \norm{\sum_{\lambda , \norm{\lambda}_1 > 2t} d_\lambda \exp(-\sigma k_\lambda) \chi_\lambda(x)}_2^2  \\
    &= \int_G d\mu(x) \abs{\sum_{\lambda , \norm{\lambda}_1 > 2t} d_\lambda \exp(-\sigma k_\lambda) \chi_\lambda(x)}^2\\ 
    &= \sum_{\nu , \norm{\nu}_1 > 2t} \sum_{\lambda , \norm{\lambda}_1 > 2t} d_\lambda d_\nu \exp(-\sigma (k_\lambda + k_\nu) )\int_G d\mu(x)  \chi^*_\nu(x)\chi_\lambda(x)\\ 
    &= \sum_{\lambda , \norm{\lambda}_1 > 2t} d_\lambda^2 \exp(-2 \sigma k_\lambda),\label{eqn:2-norm-trimmed-difference-bound}
\end{align}
using the orthonormality of the characters. The weights $\lambda$ are integer-valued $d$ dimensional vectors with non-increasing entries, which satisfy the condition $\sum_j \lambda_j=0$. We now need a bound for $\alpha(j)$, the number of highest weights $\lambda$ satisfying $\norm{\lambda}_1 = j$. Each highest weight is uniquely determined by an integer-valued $d-1$ dimensional vector with non-increasing entries, since e.g. the last element of the vector is fixed by the constraint $\sum_j \lambda_j=0$. In order to simplify the reasoning we will ignore the non-increasing property and obtain a slightly looser bound than we would by including it. Such $d-1$ dimensional vector clearly has 1-norm less than $||\lambda||_1$. Since the infinity-norm lower bounds the one-norm it follows that the number of highest-weight vectors with $1$-norm equal to $j$ is upper bounded by the number of integer vectors with infinity norm less than $j$, which is exactly the number of integer points in an $d-1$ dimensional hypercube of side length $2j$. We therefore obtain the very crude upper bound
\begin{align}
    \alpha(j) \leq (1+2j)^{d-1}.
\end{align}
Substituting this, along with the bounds from~\eqref{eqn:irrep-dimension-bound} and~\eqref{eqn:irrep-casimir-bound} into~\eqref{eqn:2-norm-trimmed-difference-bound} we obtain
\begin{align}
    \norm{H_P(\cdot, \sigma) - H_P^{(t)}(\cdot,\sigma)}_2^2 &=  \sum_{j > 2t} \alpha(j) d_\lambda^2  \exp(-2 \sigma k_\lambda)\\
    &\leq \sum_{j > 2t}(1+2j)^{(d-1)} (1+j)^{d(d-1)} \exp\left(-2 \sigma \left(\frac{j^2}{2d^2}  + \frac{j}{4}\right)\right)\\
    &\leq \left(2+\frac{1}{2t}\right)^{d-1} \left(1+\frac{1}{2t}\right)^{d(d-1)}  \sum_{j > 2t} j^{(d+1)(d-1)}\exp\left(-2 \sigma \left(\frac{j^2}{2d^2}  + \frac{j}{4}\right)\right).
\end{align}
Our approach is to bound the expression by a Gaussian integral, so we first bound the polynomial term in the sum by a Gaussian. An easy bound follows from the bound on the Lambert $W$ function
\begin{align}
    W_{-1}(-e^{u-1}) > -1 - \sqrt{2u} - u,
\end{align}
obtained in Ref.~\cite{Chatzigeorgiou-2013-lambert-w}, namely that
\begin{align}
    j^2 \geq \frac{ d^4}{\gamma\sigma}\log\left(\frac{d^4}{2\gamma\sigma}\right)  \implies (1+2j)^{(d+1)(d-1)} \leq \exp\left(\gamma \sigma \frac{j^2}{d^2} \right).
\end{align}
Assuming that
\begin{align}
    j > 2t \geq \frac{d^2}{\sqrt{\gamma\sigma}}\sqrt{\log\left(\frac{d^4}{2\gamma\sigma}\right)}  \implies (1+2j)^{(d+1)(d-1)} \leq \exp\left(\gamma \sigma \frac{j^2}{d^2} \right),
\end{align}
where $0<\gamma<1$ is a constant we have introduced. We substitute this bound to obtain 
\begin{align}
    \norm{H_P(\cdot, \sigma) - H_P^{(t)}(\cdot,\sigma)}_2^2 \leq \left(2+\frac{1}{2t}\right)^{d-1} \left(1+\frac{1}{2t}\right)^{d(d-1)} \sum_{j > 2t}\exp\left(- \sigma \left((1-\gamma)\frac{j^2}{d^2}  + \frac{j}{2}\right)\right).
\end{align}
The summand is decreasing in $j$, so we may bound the sum by an integral to obtain
\begin{align}
    \norm{H_P(\cdot, \sigma) - H_P^{(t)}(\cdot,\sigma)}_2^2 &\leq \left(2+\frac{1}{2t}\right)^{d-1} \left(1+\frac{1}{2t}\right)^{d(d-1)} \int_{2t}^\infty\exp\left(- \sigma \left((1-\gamma)\frac{x^2}{d^2}  + \frac{x}{2}\right)\right)dx.\\
\end{align}
We compute the integral and employ the standard bound
\begin{align}
    \erfc(x) &\leq \frac{1}{x\sqrt{\pi}}\exp\left(-x^2\right),
\end{align}
to obtain
\begin{align}
    \norm{H_P(\cdot, \sigma) - H_P^{(t)}(\cdot,\sigma)}_2^2 &\leq \left(2+\frac{1}{2t}\right)^{d-1} \left(1+\frac{1}{2t}\right)^{d(d-1)} \frac{2 d^2}{\sigma} \frac{1}{d^2+ 8t(1-\gamma)} \exp\left(-4\sigma(1-\gamma)\frac{t^2}{d^2} - \sigma t\right)\mathrm{.}
\end{align}
Recalling that we are still assuming that $2t \geq \frac{d^2}{\sqrt{\gamma\sigma}}\sqrt{\log\left(\frac{d^4}{2\gamma\sigma}\right)}$ we can obtain the very simple bound
\begin{align}
    \norm{H_P(\cdot, \sigma) - H_P^{(t)}(\cdot,\sigma)}_2^2 &\leq 2^d \frac{e }{1+d^2}\frac{1}{5 - 4\gamma}\exp\left(-4\sigma(1-\gamma)\frac{t^2}{d^2} - \sigma t\right)\\
    &\leq 2^d\exp\left(-4\sigma(1-\gamma)\frac{t^2}{d^2} - \sigma t\right).
\end{align}

\end{proof}

\section[Bounding the dominant term I0]{Bounding the dominant term \texorpdfstring{$\mathcal{I}_0$}{}}
\label{app:I0}
In this Appendix, we provide a proof of Lemma \ref{lem:I0_bound}, which bounds the $\mathcal{I}_0$ term~\eqref{eq:I0}. To do that, we use the following Lemmas \ref{lem:GUE0}, \ref{lem:szarek} and \ref{lem:C_value}.

\begin{lem}
\label{lem:GUE0}
Let $\mathrm{GUE}_d^0$ denote a GUE ensemble of traceless $d \times d$ matrices. Then for any $r \geq 0$ (see e.g. Ref.~\cite{Tracy-Widom-2001})

\begin{align}
    \prob_{A\sim \mathrm{GUE}_d^0}\left(\norm{A}_\infty \leq r\right) = \left(\prod_{j=1}^d\frac{1}{j!}\right) (2\pi)^{-\frac{d-1}{2}}2^{\frac{d^2-1}{2}}\int_{\mathcal{Z}^{d-1} \cap \mathcal{H}_r^d}dy\exp\left(-\sum_{j=1}^d y_j^2\right)\prod_{1\leq i < j\leq d}(y_i - y_j)^2 \mathrm{.}\label{eqn:gue-integral}
\end{align}
\end{lem}

\begin{lem}
\label{lem:szarek}
    For any $r \geq 2 \sqrt{d}$ (see Ref.~\cite{abmb-2017})
    \begin{equation}
    \prob_{A\sim \mathrm{GUE}_d^0}\left(\norm{A}_\infty \geq r\right) \leq \frac{1}{2} \exp\left(- \frac{d}{2} (\frac{r}{\sqrt{d}}-2)^2\right) \mathrm{.}
    \label{eq:szarekm}
\end{equation}
\end{lem}

\begin{proof}
    From Ref.~\cite{abmb-2017} we have
    \begin{equation}
    \prob_{A\sim \mathrm{GUE}_d^0}\left(\frac{1}{\sqrt{d}}\norm{A}_\infty \geq 2+x\right) \leq \frac{1}{2} \exp\left(- \frac{d x^2}{2}\right)
\end{equation}
valid for $x \geq 0$. The result follows via $x=\frac{r}{\sqrt{d}}-2$.
\end{proof}

\begin{lem}
\label{lem:C_value}

For $\mathrm{SU}(d)$ we can evaluate (see~\eqref{eq:C_def})
\begin{equation}
   \frac{C(d,\sigma) }{|W|} =  \frac{\sqrt{d} (2d)^{(d-1)/2+m}}{\prod_{k=1}^{d}k!} (2\pi)^{d-1+m} e^{\frac{d^2-1}{24} \sigma} (4 \pi \sigma)^{-(d^2-1)/2} \mathrm{,}
\end{equation}
where $m=d(d-1)/2$.

\end{lem}

\begin{proof}
From Ref.~\cite{urakawa-1974} we know that
$c=2^{l/2}\frac{\sqrt{D}}{\prod_{i=1}^{l} |\alpha_i|}$, where $D$ is the determinant of the Cartan matrix. It is known that for $A_{d-1}$ root system, $D=d$ and $|W|=d!$. Moreover, we have that
$N=d^2-1$, $l=d-1$, $m=d(d-1)/2$, $|\alpha_i|=1/\sqrt{d}$.
The expression $\pi(\delta)$ can be calculated from equation~\eqref{eq:delta} as

\begin{equation}
 \pi(\delta)=(2d)^{- m} i^m \prod_{1\leq p<q\leq d}(q-p)=(2d)^{- m} i^m \prod_{k=1}^{d-1} k!.  
\end{equation}

\end{proof}

\begin{lem}
\label{lem:I0_bound}
Assume $\sigma \leq \frac{\tilde{\epsilon}^2}{32}$. Then
    \begin{equation}
     \mathcal{I}_0 \leq \overline{\mathcal{I}_0} \coloneqq \frac{1}{2} \exp\left(-\frac{d}{16 \sigma} \tilde{\epsilon}^2 + \frac{d^2-1}{24} \sigma \right) \mathrm{.}
    \end{equation}
\end{lem}

\begin{proof}
    
The $\mathcal{I}_0$ can be bounded as 
\begin{align}
\label{eq:zero_order}
   \mathcal{I}_0\leq& \frac{C(d,\sigma) 2^m}{|W|} \int_{\mathcal{H}^{d-1}_{\pi} \setminus \mathcal{H}^{d-1}_{\tilde{\epsilon}}}d\mu(\phi) \abs{\prod_{\alpha > 0}\alpha(X_\phi)\sin\left(\frac{\alpha(X_{\phi})}{2}\right)} \exp\left(\frac{-1}{4\sigma}\norm{X_\phi}^2\right)\\
   \leq &\frac{ C(d,\sigma)}{|W|}\int_{(\mathcal{H}^{d-1}_{\tilde{\epsilon}})^c}d\mu(\phi) \left(\prod_{\alpha > 0}\alpha(X_\phi)^2\right) \exp\left(\frac{-1}{4\sigma}\norm{X_\phi}^2\right).\label{eqn:almost-gue-integral}
\end{align}
Recalling~\eqref{eq:killing_form} and~\eqref{eqn:maximal-torus-parameterisation} we can write
\begin{align}
    \norm{X_\phi}^2 &= 2d\left(\sum_{j=1}^{d-1}\phi_j^2 + \left(\sum_{j=1}^{d-1}\phi_j\right)^2\right),
\end{align}
and it is convenient to introduce the variables $y_j \coloneqq \phi_j\sqrt{\frac{d}{2\sigma}}$ and $y_d \coloneqq -\sum_{j=1}^{d-1}y_j$. The expression in~\eqref{eqn:almost-gue-integral} is then equal to
\begin{align}
    &\frac{ C(d,\sigma)}{|W|} \left(\frac{2\sigma}{d}\right)^{m+\frac{l}{2}} (2\pi)^{-(d-1)}\frac{1}{\sqrt{d}}\int_{\mathcal{Z}^{d-1} \cap (\mathcal{H}^d_{\tilde{\epsilon}\sqrt{\frac{d}{2\sigma}}})^c}d\mu_\mathcal{Z}(y) \left(\prod_{1\leq i <j\leq d }(y_i - y_j)^2\right) \exp\left(-\sum_{j=1}^d y_j^2\right)\label{eqn:almost-gue-integral-rescaled-1}\\
    &= e^{\norm{\delta}^2 \sigma} \left(\prod_{k=1}^d\frac{1}{k!}\right) 2^{\frac{1}{2} (d-1) d} \pi ^{\frac{1}{2}-\frac{d}{2}}\int_{\mathcal{Z}^{d-1} \cap (\mathcal{H}^d_{\tilde{\epsilon}\sqrt{\frac{d}{2\sigma}}})^c}d\mu_\mathcal{Z}(y) \left(\prod_{1\leq i <j\leq d }(y_i - y_j)^2\right) \exp\left(-\sum_{j=1}^d y_j^2\right)\label{eqn:almost-gue-integral-rescaled-2},
\end{align}
where $\mu_\mathcal{Z}$ is the Euclidean measure on the hyperplane $\mathcal{Z}^{d-1}$ and a factor of $d^{-\frac{1}{2}}$ appears in~\eqref{eqn:almost-gue-integral-rescaled-1} from changing the measure from the Euclidean one on the $d-1$ variables $y_1\hdots y_{d-1}$ to the Euclidean measure intrinsic to the hyperplane. 
The transition from~\eqref{eqn:almost-gue-integral-rescaled-1} to~\eqref{eqn:almost-gue-integral-rescaled-2} comes from the application of Lemma \ref{lem:C_value}.

Using the normalisation of the probability to 1, we can apply Lemma \ref{lem:GUE0} with $r=\sqrt{\frac{d}{2 \sigma}} \tilde{\epsilon}$ to~\eqref{eqn:almost-gue-integral-rescaled-2} to write 
\begin{align}
\mathcal{I}_{0} \leq e^{\norm{\delta}^2 \sigma} \prob_{A\sim \mathrm{GUE}_d^0}\left(\norm{A}_\infty \geq \tilde{\epsilon} \sqrt{\frac{d}{2\sigma}}\right) 
\label{eq:guebound}  \mathrm{.}
\end{align}
Applying Lemma \ref{lem:szarek} with $r=\sqrt{\frac{d}{2 \sigma}} \tilde{\epsilon}$ to~\eqref{eq:guebound} and assuming

\begin{equation}
\tilde{\epsilon} \geq \frac{2 \sqrt{2 \sigma}}{1-\beta} \mathrm{,}
\label{eq:cond_szarek_beta}
\end{equation} for some $0 < \beta<1$ (which is stronger than the assumption in Lemma \ref{lem:szarek}) we obtain the following bound

\begin{align}
\mathcal{I}_{0}&\leq \frac{1}{2} \exp\left(-\frac{d}{4 \sigma} \tilde{\epsilon}^2 + \frac{4d}{2\sqrt{2\sigma}}\tilde{\epsilon} -2d + \norm{\delta}^2 \sigma \right)\\ 
&\leq   \frac{1}{2} \exp\left(-\frac{d \beta^2}{4 \sigma} \tilde{\epsilon}^2 + \norm{\delta}^2 \sigma \right)\\
&= \frac{1}{2} \exp\left(-\frac{d\beta^2}{4 \sigma} \tilde{\epsilon}^2 + \frac{d^2-1}{24} \sigma \right)  \mathrm{.}
\label{eq:0_order_bound}
\end{align}
We set $\beta=1/2$ and for $\sigma \leq \frac{\tilde{\epsilon}^2}{32}$~\eqref{eq:cond_szarek_beta} we get

\begin{equation}
\label{eq:I0b}
   \mathcal{I}_0 \leq \frac{1}{2} \exp\left(-\frac{d}{16 \sigma} \tilde{\epsilon}^2 + \frac{d^2-1}{24} \sigma \right)  \mathrm{.}
\end{equation}
\end{proof}

\section{Bounding the correction term \texorpdfstring{$\mathcal{R}$}{}}
\label{app:R}

In this Appendix, we bound the remaining terms $\mathcal{R}$, defined in~\eqref{eq:R}. Our strategy is to bound $\mathcal{R}$ by the volume of the complement of $\epsilon$-ball times the upper bound on the integrand outside of the ball (Lemma \ref{lem:R_bound}). We then show that $\mathcal{R}$ is indeed a correction to $\mathcal{I}_0$ (see~\eqref{eq:I0}), which relatively decays very fast with growing $d$  (Lemma \ref{lem:tail-bound-smaller-than-dominant}). To do so, we employ the following Lemmas \ref{lem:ball_points}, \ref{lem:simple_gamma_bound} and \ref{lem:Theta}.

\begin{lem}
\label{lem:ball_points}
     The number of $k$-vectors in each $|| \cdot ||_{\infty}$ norm shell of radius $r > 0$
    \begin{equation}
        S_{r,d} \coloneqq \{(n_1, \ldots, n_{d}) \in \mathbb{Z}^d | \,\,\max_{i}|n_i| = r\}
    \end{equation}
     can be upper bounded as
     \begin{equation}
         |S_{r,d}| \leq 2^{d} (2r)^{d-1}  \mathrm{.}
     \end{equation}
\end{lem}

\begin{proof}
Consider the corresponding balls   

\begin{equation}
    B_{r,d} \coloneqq \{(n_1, \ldots, n_{d}) \in \mathbb{Z}^d | \,\,\max_{i}|n_i|\leq r\}  \mathrm{.}
\end{equation}
It is easy to see that $|B_{r,d}|=(2r+1)^d$. Thus, 
\begin{equation}
    |S_{r,d}|=(2r+1)^{d}-(2r-1)^{d}  \mathrm{.}
\end{equation}
Using the binomial expansion, we can write

\begin{equation}
   |S_{r,d}| = 2 \sum_{k=0}^{\floor{\frac{d-1}{2}}} \binom{d}{2k+1} (2r)^{d-2k-1} \leq 2 (2r)^{d-1} \sum_{k=0}^{\floor{\frac{d-1}{2}}} \binom{d}{2k+1} = 2^{d} (2r)^{d-1}  \mathrm{.}
\end{equation}
\end{proof}

\begin{lem}  
\label{lem:simple_gamma_bound}
Let $\Gamma(s,x)$ denote the upper incomplete Gamma function

\begin{equation}
    \Gamma(s,x) \coloneqq \int_x^{\infty} t^{s-1}e^{-t} dt \mathrm{.} 
\end{equation}
Then, assuming $s \geq 1$ and $x > s-1$

\begin{equation}
    \Gamma(s,x) \leq \frac{e^{-x} x^s}{x-s+1} \mathrm{.}
\end{equation}

\end{lem}

\begin{proof}
   \begin{equation}
    \Gamma(s, x ) = e^{-x} \int_0^{\infty} (t+x)^{s-1}e^{-t} dt \leq e^{-x} x^{s-1} \int_0^{\infty} e^{\frac{t}{x} (s-1) -t} dt
\end{equation} 
This bound can be improved using continued fraction representation.
\end{proof}

\begin{lem}
\label{lem:Theta}
\begin{equation}
  \left(\frac{d}{4} \right)^{-d^2/8} \geq \frac{1}{\prod_{k=1}^{d}k!}
\end{equation}   
\end{lem}

\begin{proof}
We first lower bound the value of $\mathrm{log}(k!)$ from below. It is clear that a sum
\begin{equation}
    \mathrm{log}(k!) = \mathrm{log}(1) + \mathrm{log}(2) + \ldots + \mathrm{log}(k)
\end{equation}
can be lower bounded by $(k-j+1)\mathrm{log}(j)$ for any $1 \leq j \leq k$. Picking $j=\ceil{\frac{k}{2}}$ and using the monotonicity of $x \,\mathrm{log}(x)$ we obtain 
\begin{equation}
    \frac{k}{2}\mathrm{log}\left( \frac{k}{2}\right) \leq \mathrm{log}(k!) \mathrm{.}
\end{equation}
Using this bound and repeating the argument for
\begin{equation}
    \mathrm{log}\left(\prod_{k=1}^{d}k!\right) = \mathrm{log}(1!) + \mathrm{log}(2!) + \ldots + \mathrm{log}(d!)
\end{equation}
we obtain
\begin{equation}
     \frac{d^2}{8} \mathrm{log}\left( \frac{d}{4} \right)\leq \mathrm{log}\left(\prod_{k=1}^{d}k!\right)  \mathrm{.}
\end{equation}
The result follows via exponentiation.

\end{proof}

\begin{lem}
\label{lem:R_bound}
    \begin{equation}
        \mathcal{R} \leq \overline{\mathcal{R}} \coloneqq \frac{C(d,\sigma) 2^m}{|W|} 2^{d-1}  \cdot  (2\pi)^{m} \left(1+\frac{d}{2}\right)^m  2^{m+d-1}e^{-\frac{d \pi^2}{2 \sigma}}
    \end{equation}
    for $\sigma \leq \frac{2 \pi^2 d}{d^2+d-2}$.
\end{lem}

\begin{proof}
    
Consider a summand for some fixed $k\neq 0$
\begin{align}
 &\frac{ C(d,\sigma) 2^m}{|W|} \int_{\mathcal{H}^{d-1}_{\pi} \setminus \mathcal{H}^{d-1}_{\tilde{\epsilon}}} d\mu(\phi) \abs{\prod_{\alpha > 0}\alpha(X_\phi + X_k)\sin\left(\frac{\alpha(X_{\phi})}{2}\right)}  \exp\left(-\frac{1}{4\sigma} \norm{X_\phi + X_k }^2\right) \\
 \leq & \frac{ C(d,\sigma) 2^m}{|W|}\int_{\mathcal{H}^{d-1}_{\pi} \setminus \mathcal{H}^{d-1}_{\tilde{\epsilon}}} d\mu(\phi) \abs{ \prod_{\alpha > 0}\alpha(X_\phi+X_k) } \exp\left(\frac{-1}{4\sigma}\norm{X_\phi + X_k}^2\right).
\end{align}
We have
\begin{equation}
  \prod_{\alpha > 0}\alpha(X_\phi + X_k) = \prod_{1 \leq i < j < d} \left(\phi_i - \phi_j + 2\pi(k_i-k_j)\right) \prod_{1 \leq i \leq d-1}(\phi_i - \phi_d + 2 \pi \left(k_i-k_d)\right)
\end{equation}
and on the domain of integration we can bound $\abs{\phi_i - \phi_j} \leq 2 \pi$  and $\abs{k_i-k_j} \leq d ||k||_{\infty}$ for $i<j\leq d$. Thus,

\begin{align}
   \abs{\prod_{\alpha > 0}\alpha(X_\phi + X_k)} &\leq \ \left( 2\pi + 2\pi d ||k||_{\infty}\right)^m \\ &=  (2\pi)^m (1+d||k||_{\infty})^m  \mathrm{.}
\end{align}

Let us find lower bounds on the exponents. We have

\begin{equation}
\label{eq:sq_sum}
    ||X_{\phi} + X_k ||^2 = 2d \left(\sum_{i=1}^{d-1}(\phi_i+2\pi k_i )^2 + (\phi_d + 2\pi k_d )^2\right)  \mathrm{.}
\end{equation}

The first term of~\eqref{eq:sq_sum} is just the square of the Euclidean distance from the origin in coordinate space. This way, Fig. \ref{fig:summation} can be used to understand the summation and bounding process better. The second term is non-negative and is a square of the sum of all coordinates. As such, it has a minimum of zero on the hyperplane crossing the origin. To simplify the reasoning, we discard the second term altogether and obtain the isotropic bound

\begin{align}
    ||X_{\phi} + X_k ||^2 &\geq 2d \left(2\pi ||k||_{\infty} -\pi\right)^2\\
    &= 2d \pi^2 (2||k||_{\infty}-1)^2 \mathrm{.}
\end{align}
Denoting

\begin{equation}
    \Psi(\phi, k) \coloneqq \int_{\mathcal{H}^{d-1}_{\pi} \setminus \mathcal{H}^{d-1}_{\tilde{\epsilon}}} d\mu(\phi) \abs{\prod_{\alpha > 0}\alpha(X_\phi + X_k)\sin\left(\frac{\alpha(X_{\phi})}{2}\right)} \exp\left(-\frac{1}{4\sigma} \norm{X_\phi + X_k }^2\right)  
\end{equation}
and assuming $\sigma$ is small enough, namely

\begin{equation}
\label{eq:smallA}
    \sigma \leq \frac{2 \pi^2 d}{d^2+d-2}\mathrm{,}
\end{equation}
we can use Lemmas \ref{lem:ball_points} and \ref{lem:simple_gamma_bound} to bound the correction term $\mathcal{R}$ as follows \footnote{Below we slightly abused the notation by denoting the $||k||_{\infty}$ as $k$.}

\begin{align}
\mathcal{R} &\leq \frac {C(d,\sigma) 2^m}{|W|} \sum_{k\neq 0}    \Psi(\phi, k) \\ &\leq \frac {C(d,\sigma) 2^m}{|W|} \sum_{k\neq 0}   \cdot (1-\mathrm{Vol}(B_{\epsilon})) \sum_{k=1}^{\infty}|S_{k,d-1}| (2\pi)^{m} (1+d k)^m \exp\left( - \frac{d \pi^2 }{2 \sigma} (2k-1)^2 \right) \label{eq:init} \\
& \leq \frac{C(d,\sigma) 2^m}{|W|}   \cdot  2^{2(d-2)+1}  (2\pi)^{m} \sum_{k=1}^{\infty} k^{d-2}  (1+d k)^m \exp\left( - \frac{d \pi^2 }{2 \sigma} (2k-1)^2 \right) \label{eq:dropped_v} \\
& = \frac{C(d,\sigma) 2^m}{|W|} 2^{d-1}  (2\pi)^{m}  \sum_{u=1, u \mathrm{\, odd}}^{\infty}(1+u)^{d-2}   \left(1+\frac{d}{2}(u+1)\right)^m \exp\left( - \frac{d \pi^2 }{2 \sigma} u^2 \right) \label{eq:u_subs} \\
& = \frac{C(d,\sigma) 2^m}{|W|}  2^{d-1}    (2\pi)^{m} \left(1+\frac{d}{2}\right)^m \sum_{u=1, u \mathrm{\, odd}}^{\infty}(1+u)^{d-2}   \left(1+\frac{d}{d+2} u \right)^m \exp\left( - \frac{d \pi^2 }{2 \sigma} u^2 \right) \label{eq:u_ugly} \\
 & \leq \frac{C(d,\sigma) 2^m}{|W|} 2^{d-1}   \cdot  (2\pi)^{m} \left(1+\frac{d}{2}\right)^m  2^{m+d-2}\sum_{u=1, u \mathrm{\, odd}}^{\infty} u^{m+d-2} \exp\left( - \frac{d \pi^2 }{2 \sigma} u^2 \right) \label{eq:u_odd}\\
 & \leq \frac{C(d,\sigma) 2^m}{|W|} 2^{d-1}  \cdot (2\pi)^{m} \left(1+\frac{d}{2}\right)^m  2^{m+d-2} \left(  e^{-\frac{d \pi^2}{2 \sigma}} + \frac{1}{2}\int_1^{\infty}u^{m+d-2} \exp\left( - \frac{d \pi^2 }{2 \sigma} u^2 \right) d u \right)  \label{eq:u_int} \\
  & = \frac{C(d,\sigma) 2^m}{|W|} 2^{d-1}  \cdot  (2\pi)^{m} \left(1+\frac{d}{2}\right)^m  2^{m+d-2} \left( e^{-\frac{d \pi^2}{2 \sigma}} + \frac{1}{2^2 (\frac{d \pi^2}{2 \sigma})^{\frac{m+d-1}{2}}} \Gamma\left(\frac{m+d-1}{2},  \frac{d \pi^2}{2 \sigma}\right) \right) \label{eq:gamma} \\
  &=\frac{C(d,\sigma) 2^m}{|W|} 2^{d-1}  \cdot  (2\pi)^{m} \left(1+\frac{d}{2}\right)^m  2^{m+d-2} \left( 1 + \frac{1}{4(A(d, \sigma)+1)}\right) e^{-\frac{d \pi^2}{2 \sigma}} \label{eq:A} \\
  & \leq  \frac{C(d,\sigma) 2^m}{|W|} 2^{d-1}  \cdot  (2\pi)^{m} \left(1+\frac{d}{2}\right)^m  2^{m+d-1}e^{-\frac{d \pi^2}{2 \sigma}} \label{eq:fin}  \mathrm{,}
\end{align}
where

\begin{equation}
\label{eq:A_def}
    A(d, \sigma) \coloneqq \frac{d \pi^2}{2 \sigma} - \frac{m+d-1}{2} \mathrm{.}
\end{equation}
Let us explain the bounding process in more detail. We applied the bound $1-\mathrm{Vol}(B_{\epsilon}) \leq 1 $ and Lemma \ref{lem:ball_points} to bound~\eqref{eq:init} by~\eqref{eq:dropped_v}. We substituted $u=2k-1$ to~\eqref{eq:dropped_v}. We bounded $d/(d+2) \leq 1$ \footnote{Otherwise, $u=1$ needs to be considered separately when bounding binomial by the power of two, since $du/(d+2) < 1$. for $u=1$.} and applied a very crude bound 

\begin{equation}
    (1+u)^{m+d-2} \leq 2^{m+d-2} \cdot u^{m+d-2}
\end{equation} to bound~\eqref{eq:u_ugly} by~\eqref{eq:u_odd}.
The function
\begin{equation}
    f(u) \coloneqq u^{m+d-1} e^{-\frac{d \pi^2}{2 \sigma} u^2}
\end{equation}
is increasing from 0 at $u=0$ to its local maximum and then is decreasing. The condition~\eqref{eq:smallA} guarantees that the local maximum of $f(u)$ for $u > 0$ is smaller or equal to $1$, since this requires a weaker condition
\begin{equation}
    \sigma \leq \frac{d \pi^2}{m+d-2} \mathrm{.}
\end{equation}
This requirement allows us to bound the sum over odd $u$~\eqref{eq:u_odd} in terms of its first term plus an appropriate integral~\eqref{eq:u_int}.
We applied a well-known formula
\begin{equation}
    \int_{a}^{\infty} x^d e^{-\alpha x^2} dx = \frac{\Gamma(\frac{d+1}{2}, a^2 \alpha)}{2 \alpha^{\frac{d+1}{2}}}
\end{equation}
to~\eqref{eq:u_int}. Finally, due to~\eqref{eq:smallA}, we have that  $A(d, \sigma) \geq 0$  which allows us to apply Lemma \ref{lem:simple_gamma_bound} \footnote{The application of Lemma \ref{lem:simple_gamma_bound} requires $A(d, \sigma) > -1$.} to bound~\eqref{eq:gamma} by~\eqref{eq:A} and bound~\eqref{eq:A} by~\eqref{eq:fin}.  
\end{proof}
We want to compare the upper bound on $\mathcal{R}$ from Lemma \ref{lem:R_bound} with an upper bound on $\mathcal{I}_0$ from Lemma \ref{lem:I0_bound}.

\begin{lem}\label{lem:tail-bound-smaller-than-dominant}
    Let $\sigma \leq \frac{1}{d \mathrm{log}(d)}$. Then,
    \begin{equation}
        \overline{\mathcal{R}} \leq \cdot \eta  \overline{\mathcal{I}_0}
    \end{equation}
    for $\eta \geq \frac{1}{\prod_{k=1}^{d}k!}$ and $d\geq 2$. In particular, we can take $\eta \geq 1/2$ to obtain a uniform bound for all $d \geq 2$.
\end{lem}

\begin{proof}
Clearly, from Lemma \ref{lem:I0_bound}
\begin{equation}
    \overline{\mathcal{I}}_0 \geq \frac{1}{2} \exp\left(-\frac{d \pi^2}{16 \sigma}  + \frac{d^2-1}{24} \sigma \right) \mathrm{,}
\end{equation}
so that using Lemmas \ref{lem:C_value} and \ref{lem:R_bound}

\begin{equation}
\label{eq:ratio}
    \frac{\overline{\mathcal{R}}}{\overline{\mathcal{I}}_0} \leq \frac{R(d)}{\prod_{k=1}^{d}k!} e^{-\frac{7}{16}\frac{d \pi^2}{ \sigma}} \mathrm{,}
\end{equation}
where

\begin{equation}
\label{eq:R_d}
    R(d) \coloneqq 2^{\frac{7d}{2}-d^2+4m-\frac{3}{2}} \cdot \pi^{d-\frac{d^2}{2}+2m-\frac{1}{2}} d^{\frac{d}{2}+m}(d+2)^m \mathrm{.}
\end{equation}
Demanding the ratio bound~\eqref{eq:ratio} to be smaller than $\eta$ yields

\begin{equation}
\label{eq:s_bound}
    \sigma \leq \frac{7 d \pi^2}{16} \frac{1}{\mathrm{log}(R)-\mathrm{log}(\prod_{k=1}^{d}k!) - \mathrm{log}(\eta)} \mathrm{.}
\end{equation}
Bounding $\mathrm{log}(R) \leq 3d^2 \mathrm{log}(d)$ (for $d \geq 2$)  we obtain

\begin{equation}
   \frac{7 \pi^2}{48 d \mathrm{log}(d)} \leq \frac{7 d \pi^2}{16} \frac{1}{\mathrm{log}(R)} \leq \frac{7 d \pi^2}{16} \frac{1}{\mathrm{log}(R)-\mathrm{log}(\prod_{k=1}^{d}k!) - \mathrm{log}(\eta)} \mathrm{,}
\end{equation}
so~\eqref{eq:s_bound} is satisfied for
\begin{equation}
\label{eq:smalltail}
    \sigma \leq \frac{1}{ d \mathrm{log}(d)}
\end{equation}
and $\eta \geq \frac{1}{\prod_{k=1}^{d}k!}$, with the right-hand side decaying very fast with $d$ and upper bounded using Lemma \ref{lem:Theta}.

Note that~\eqref{eq:smalltail} is stronger than~\eqref{eq:smallA}.

\end{proof}

\section{Bounding the \texorpdfstring{$L^2$}{L2}-norm}
\label{app:L2}
In this Appendix, we prove Lemma \ref{lem:L2_bound} and Corollary \ref{lem:L2_bound2}, which bound the $L^2$-norm of the trimmed heat kernel $H_P^{(t)}$. The $L^2$-norm is divided into two contributions (see~\eqref{eq:l2_split}) which are bounded separately in Lemmas \ref{lem:I^2} and \ref{lem:R^2}.

\begin{lem}
 \label{lem:I^2}
 \begin{equation}
      \mathcal{I}_{0,0}^2 = \frac{C(d, \sigma)}{2^{m+\frac{l}{2}}} e^{\norm{\delta}^2 \sigma} \mathrm{}
 \end{equation}
 \end{lem}

 \begin{proof}
     We have
\begin{align}
   \mathcal{I}_{0,0}^2 =   \frac{C(d, \sigma)^2}{|W|} \int \left( \prod_{\alpha > 0} \alpha(X_{\phi})^2\right) \exp\left(-\frac{1}{2\sigma} \norm{X_\phi}^2 \right) d\mu(\phi) \mathrm{,}
\end{align}
which is very similar to the integral analysed previously in Lemma \ref{lem:I0_bound}. We introduce the variables $y_j = \phi_j\sqrt{\frac{d}{\sigma}}$ and $y_d = -\sum_{j=1}^{d-1}y_j$ and obtain
\begin{align}
    \mathcal{I}_{0,0}^2&=\frac{ C(d,\sigma)^2}{|W|} \left(\frac{\sigma}{d}\right)^{m+\frac{l}{2}} (2\pi)^{-(d-1)}\frac{1}{\sqrt{d}}\int_{\mathcal{Z}}d\mu_\mathcal{Z}(y) \left(\prod_{1\leq i <j\leq d }(y_i - y_j)^2\right) \exp\left(-\sum_j y_j^2\right)\\
    &=\frac{C(d, \sigma)}{2^{m+\frac{l}{2}}}\frac{ C(d,\sigma)}{|W|} \left(\frac{2\sigma}{d}\right)^{m+\frac{l}{2}} (2\pi)^{-(d-1)}\frac{1}{\sqrt{d}}\int_{\mathcal{Z}}d\mu_\mathcal{Z}(y) \left(\prod_{1\leq i <j\leq d }(y_i - y_j)^2\right) \exp\left(-\sum_j y_j^2\right)\\
    &= \frac{C(d, \sigma)}{2^{m+\frac{l}{2}}} e^{\norm{\delta}^2 \sigma} \left(\prod_{k=1}^d\frac{1}{k!}\right) 2^{\frac{1}{2} (d-1) d} \pi ^{\frac{1}{2}-\frac{d}{2}}\int_{\mathcal{Z}}d\mu_\mathcal{Z}(y) \left(\prod_{1\leq i <j\leq d }(y_i - y_j)^2\right) \exp\left(-\sum_j y_j^2\right),\\
    &= \frac{C(d, \sigma)}{2^{m+\frac{l}{2}}} e^{\norm{\delta}^2 \sigma} \prob_{A\sim \mathrm{GUE}_k^0}\left(\norm{A}_\infty \leq \infty \right) \\
    &=\frac{C(d, \sigma)}{2^{m+\frac{l}{2}}} e^{\norm{\delta}^2 \sigma} \mathrm{.}
\end{align}
 \end{proof}

 \begin{lem}
 \label{lem:R^2}
     \begin{equation}
         \mathcal{R}_{*,0}^2 + \mathcal{R}_{0,*}^2 + \mathcal{R}_{*,*}^2 \leq \frac{9}{8} \frac{d!}{2^{2m}} \overline{\mathcal{R}}^2
     \end{equation}
    for $\sigma \leq \frac{2 \pi^2 d}{d^2+d-2}$ and $d \geq 2$.
 \end{lem}

 \begin{proof}

 The proof goes along the same lines as in Lemma \ref{lem:R_bound}, so we direct the reader there for an explanation.

 Denoting

 \begin{equation}
      \Psi(\phi, k,l) \coloneqq \int_{\mathcal{H}^{d-1}_{\pi} } d\mu(\phi) \abs{\prod_{\alpha > 0}\alpha(X_\phi + X_k) \prod_{\alpha > 0}\alpha^{*}(X_\phi + X_l)} \exp\left(-\frac{1}{4\sigma} \left(\norm{X_\phi + X_k }^2+\norm{X_\phi + X_l }^2\right)  \right)
\end{equation}

we have

\begin{align}
 \mathcal{R}_{*,*}^2&=\frac {C(d,\sigma)^2 }{|W|} \sum_{k\neq 0} \sum_{l\neq 0}     \Psi(\phi, k,l) \\ &\leq \frac {C(d,\sigma)^2}{|W|}  \sum_{k=1}^{\infty}\sum_{l=1}^{\infty}|S_{k,d-1}| |S_{l,d-1}| (2\pi)^{2m} (1+d k)^m (1+d l)^m \exp\left( - \frac{d \pi^2 }{ 2\sigma} ((2k-1)^2+(2l-1)^2)  \right) \\
& \leq \frac{C(d,\sigma)^2 }{|W|}   \cdot  2^{4(d-2)+2}  (2\pi)^{2m} \left(\sum_{k=1}^{\infty}  k^{d-2}  (1+d k)^m \exp\left( - \frac{d \pi^2 }{ 2\sigma} (2k-1)^2  \right) \right)^2  \\
& = \frac{C(d,\sigma)^2 }{|W|} 2^{2(d-1)}  (2\pi)^{2m} \left( \sum_{u=1, u \mathrm{\, odd}}^{\infty}(1+u)^{d-2}    \left(1+\frac{d}{2}(u+1)\right)^m  \exp\left( - \frac{d \pi^2 }{ 2\sigma} u^2 \right) \right)^2 \\
& = \frac{C(d,\sigma)^2 }{|W|}  2^{2(d-1)}    (2\pi)^{2m} \left(1+\frac{d}{2}\right)^{2m} \left( \sum_{u=1, u \mathrm{\, odd}}^{\infty}(1+u)^{d-2}   \left(1+\frac{d}{d+2} u \right)^m  \exp\left( - \frac{d \pi^2 }{ 2\sigma} u^2 \right)\right)^2  \\
 & \leq \frac{C(d,\sigma)^2 }{|W|} 2^{2(d-1)}   \cdot  (2\pi)^{2m} \left(1+\frac{d}{2}\right)^{2m}  2^{2(m+d-2)} \left(\sum_{u=1, u \mathrm{\, odd}}^{\infty}  u^{m+d-2} \exp\left( - \frac{d \pi^2 }{ 2\sigma} u^2 \right)\right)^2 \\
 & \leq \frac{C(d,\sigma)^2 }{|W|} 2^{2(d-1)}  \cdot (2\pi)^{2m} \left(1+\frac{d}{2}\right)^{2m}   2^{2(m+d-2)} \left(  e^{-\frac{d \pi^2}{2 \sigma}} + \frac{1}{2}\int_1^{\infty}u^{m+d-2} \exp\left( - \frac{d \pi^2 }{2 \sigma} u^2 \right) d u \right)^2  \\
  & = \frac{C(d,\sigma)^2 }{|W|} 2^{2(d-1)}  \cdot  (2\pi)^{2m} \left(1+\frac{d}{2}\right)^{2m}   2^{2(m+d-2)} \left( e^{-\frac{d \pi^2}{2 \sigma}} + \frac{1}{2^2 (\frac{d \pi^2}{2 \sigma})^{\frac{m+d-1}{2}}} \Gamma\left(\frac{m+d-1}{2},  \frac{d \pi^2}{2 \sigma}\right) \right)^2  \\
  &=\frac{C(d,\sigma)^2 }{|W|} 2^{2(d-1)}  \cdot  (2\pi)^{2m} \left(1+\frac{d}{2}\right)^{2m}   2^{2(m+d-2)} \left( 1 + \frac{1}{4(A(d, \sigma)+1)}\right)^2 e^{-\frac{d \pi^2}{ \sigma}}  \\
  & \leq  \frac{C(d,\sigma)^2 }{|W|} 2^{2(d-1)}  \cdot  (2\pi)^{2m} \left(1+\frac{d}{2}\right)^{2m}  2^{2(m+d-1)}e^{-\frac{d \pi^2}{ \sigma}} \mathrm{,}
\end{align}

where $A(d,\sigma)$ is defined by~\eqref{eq:A_def}. Similarly,

\begin{align}
 \mathcal{R}_{*,0}^2&=\frac {C(d,\sigma)^2 }{|W|} \sum_{k\neq 0}     \Psi(\phi,k,0) \\ &\leq \frac {C(d,\sigma)^2}{|W|}  \sum_{k=1}^{\infty}|S_{k,d-1}|  (2\pi)^{2m} (1+d k)^m  \exp\left( - \frac{d \pi^2 }{ 2\sigma} ((2k-1)^2+1) \right)\\
& \leq \frac{C(d,\sigma)^2 }{|W|}   \cdot  2^{2(d-2)+1}  (2\pi)^{2m} e^{-\frac{d \pi^2}{2 \sigma}} \left(\sum_{k=1}^{\infty}  k^{d-2}  (1+d k)^m \exp\left( - \frac{d \pi^2 }{ 2\sigma} (2k-1)^2  \right) \right)  \\
& = \frac{C(d,\sigma)^2 }{|W|} 2^{d-1}  (2\pi)^{2m} e^{-\frac{d \pi^2}{2 \sigma}}\left( \sum_{u=1, u \mathrm{\, odd}}^{\infty}(1+u)^{d-2}    \left(1+\frac{d}{2}(u+1)\right)^m  \exp\left( - \frac{d \pi^2 }{ 2\sigma} u^2 \right) \right) \\
& = \frac{C(d,\sigma)^2 }{|W|}  2^{d-1}    (2\pi)^{2m} e^{-\frac{d \pi^2}{2 \sigma}} \left(1+\frac{d}{2}\right)^{m} \left( \sum_{u=1, u \mathrm{\, odd}}^{\infty}(1+u)^{d-2}   \left(1+\frac{d}{d+2} u \right)^m  \exp\left( - \frac{d \pi^2 }{ 2\sigma} u^2 \right)\right)  \\
 & \leq \frac{C(d,\sigma)^2 }{|W|} 2^{d-1}   \cdot  (2\pi)^{2m} e^{-\frac{d \pi^2}{2 \sigma}} \left(1+\frac{d}{2}\right)^{m}  2^{m+d-2} \left(\sum_{u=1, u \mathrm{\, odd}}^{\infty}  u^{m+d-2} \exp\left( - \frac{d \pi^2 }{ 2\sigma} u^2 \right)\right) \\
 & \leq \frac{C(d,\sigma)^2 }{|W|} 2^{d-1}  \cdot (2\pi)^{2m} e^{-\frac{d \pi^2}{2 \sigma}} \left(1+\frac{d}{2}\right)^{m}   2^{m+d-2} \left(  e^{-\frac{d \pi^2}{2 \sigma}} + \frac{1}{2}\int_1^{\infty}u^{m+d-2} \exp\left( - \frac{d \pi^2 }{2 \sigma} u^2 \right) d u \right)  \\
  & = \frac{C(d,\sigma)^2 }{|W|} 2^{d-1}  \cdot  (2\pi)^{2m} e^{-\frac{d \pi^2}{2 \sigma}} \left(1+\frac{d}{2}\right)^{m}   2^{m+d-2} \left( e^{-\frac{d \pi^2}{2 \sigma}} + \frac{1}{2^2 (\frac{d \pi^2}{2 \sigma})^{\frac{m+d-1}{2}}} \Gamma\left(\frac{m+d-1}{2},  \frac{d \pi^2}{2 \sigma}\right) \right)  \\
  &=\frac{C(d,\sigma)^2 }{|W|} 2^{d-1}  \cdot  (2\pi)^{2m} \left(1+\frac{d}{2}\right)^{m}   2^{m+d-2} \left( 1 + \frac{1}{4(A(d, \sigma)+1)}\right) e^{-\frac{d \pi^2}{ \sigma}}  \\
  & \leq  \frac{C(d,\sigma)^2 }{|W|} 2^{d-1}  \cdot  (2\pi)^{2m} \left(1+\frac{d}{2}\right)^{m}  2^{m+d-1}e^{-\frac{d \pi^2}{ \sigma}} 
\end{align}

and $\mathcal{R}_{0,*}^2=\mathcal{R}_{0,*}^2$. Thus, for $d \geq 2$ we have

\begin{align}
  \mathcal{R}_{*,0}^2 + \mathcal{R}_{0,*}^2 + \mathcal{R}_{*,*}^2 &\leq  \frac{C(d,\sigma)^2 }{|W|} 2^{2(d-1)}  \cdot  (2\pi)^{2m} \left(1+\frac{d}{2}\right)^{2m}  2^{2(m+d-1)}e^{-\frac{d \pi^2}{ \sigma}} \left( 1+\frac{2}{2^{d-1} \left(1+\frac{d}{2}\right)^{m} 2^{m+d-1}}\right)  \\
  &= \frac{d!}{2^{2m}} \overline{\mathcal{R}}^2 \left( 1+\frac{2}{2^{d-1} \left(1+\frac{d}{2}\right)^{m} 2^{m+d-1}}\right)\\
  &\leq \frac{9}{8} \frac{d!}{2^{2m}} \overline{\mathcal{R}}^2 \mathrm{.}
\end{align}
\end{proof}

\begin{lem}[Bound on the $L^2$-norm of the (trimmed) heat kernel]
\label{lem:L2_bound}
       \begin{equation}
         ||H_P^{(t)}(\cdot,\sigma)||_2 \leq d \mathcal{I}_{0,0} + d\frac{\sqrt{d!}}{2^{m-1}} \eta \overline{\mathcal{I}_{0}}
     \end{equation}
    for $\sigma \leq \frac{1}{d \mathrm{log}(d)}$, $d \geq 2$ and any $\eta \geq \frac{1}{\prod_{k=1}^d k!}$.
\end{lem}

\begin{proof}
    It is easy to see that $\norm{ H_P(\cdot,\sigma)}_2 \leq |\Gamma|\cdot||H_S(\cdot,\sigma)||_2$. Thus, from Lemmas \ref{lem:I^2} and \ref{lem:R^2}
    \begin{equation}
    \label{eq:L2_b}
        \norm{ H_P(\cdot,\sigma)}_2 \leq d \sqrt{\mathcal{I}_{0,0}^2 + \frac{9}{8} \frac{d!}{2^{2m}} \overline{\mathcal{R}}^2 } \leq d \mathcal{I}_{0,0} + \frac{3 \sqrt{2}}{4} d \frac{\sqrt{d!}}{2^{m}} \overline{\mathcal{R}}  \mathrm{.}
    \end{equation}
    Since the terms in~\eqref{eq:l2_split} are non-negative, it is clear that~\eqref{eq:L2_b} can be applied to trimmed heat kernels. 
    The result follows from bounding $\frac{3 \sqrt{2}}{4} \leq 2$ and the application of Lemma \ref{lem:tail-bound-smaller-than-dominant}.
\end{proof}

\begin{cor}
\label{lem:L2_bound2}
       \begin{equation}
         ||H_P^{(t)}(\cdot,\sigma)||_2 \leq c \left(\frac{d}{\sigma}\right)^{\frac{d^2-1}{4}}
     \end{equation}
    for $\sigma \leq \frac{1}{d \mathrm{log}(d)}$ and $d \geq 2$, where $c$ is some positive group constant. For example, one can take $c = 8$ for $d\geq 2$ and  $c=1$ for $d\geq 12$. 
\end{cor}

\begin{proof}

Using the proof of Lemma \ref{lem:L2_bound}

    \begin{align}
  \norm{ H_P(\cdot,\sigma)}_2 &\leq   d^{\frac{3}{16}d^2+1} \sqrt{d!} 2^{-\frac{d^2}{8}+\frac{d}{4}} \pi^{\frac{d-1}{4}} e^{\frac{d^2-1}{24} \sigma} \sigma^{-\frac{d^2-1}{4}} + \frac{3\sqrt{2}}{4} d\frac{\sqrt{d!}}{2^{m}}\frac{1}{\prod_{k=1}^d k!} e^{ \frac{d^2-1}{24} \sigma } \\
  &\leq   d^{\frac{3}{16}d^2+1} \sqrt{d!} 2^{-\frac{d^2}{8}+\frac{d}{4}} \pi^{\frac{d-1}{4}} e^{\frac{d^2-1}{24} \sigma} \sigma^{-\frac{d^2-1}{4}} + \frac{3\sqrt{2}}{4} d\frac{\sqrt{d!}}{2^{m}}\left(\frac{d}{4} \right)^{-d^2/8} e^{ \frac{d^2-1}{24} \sigma } \\
  &=  d\frac{\sqrt{d!}}{2^{m}}\left(\frac{d}{4} \right)^{-d^2/8} e^{ \frac{d^2-1}{24} \sigma } \left(\frac{3\sqrt{2}}{4}+d^{\frac{5}{16}d^2}  2^{d^2/8-d/4} \pi^{\frac{d-1}{4}}  \sigma^{-\frac{d^2-1}{4}}\right)\\
  &\leq  d^{\frac{3}{16}d^2+1} \sqrt{d!} 2^{-\frac{d^2}{8}+\frac{d}{4}+1} \pi^{\frac{d-1}{4}} e^{\frac{d^2-1}{24} \frac{1}{d \mathrm{log}(d)}} \sigma^{-\frac{d^2-1}{4}}
\end{align}
for $\sigma \leq \frac{1}{d \mathrm{log}(d)}$. The logarithm of the sigma-independent terms can be upper bounded by $\frac{1}{4} (d^2-1)\mathrm{log}(d)$ for $d\geq 8$. For $d\geq 2$, it can be upper bounded by $ \frac{1}{4} (d^2-1)\mathrm{log}(d)+ \mathrm{log}(19)$.
    
\end{proof}

\section{Well-definedness of the Poisson form of the heat kernel at non-regular points}
\label{sec:well-defined-heat-kernel}
A complete proof of this is beyond the scope of this manuscript, and perhaps the easiest such proof is that of Urakawa~\cite{urakawa-1974}, showing that the two forms of the heat kernel are equivalent. In this example, we will sketch the idea by demonstrating the phenomenon in the limit as two elements of the vector $\phi$ become equal to each other. 

Fix all elements of $\phi$ in equation~\eqref{eqn:heat-kernel-poisson-specialised} other than $\phi_1$ and $\phi_2$, set these to be $a+b$ and $a-b$, respectively. Assume all of the $\phi_j$ are not equal to each other, and none are in the interval $(a-b, a+b)$. We will consider the limit as $ b \to 0$. The relevant part of~\eqref{eqn:heat-kernel-poisson-specialised} becomes
\begin{align}
    j(\exp(X_\phi))^{-1} &\sum_{k\in\mathbb{Z}^{d-1}} \pi(X_\phi + X_k) \exp\left(-\frac{1}{4\sigma} \norm{X_\phi + X_k}^2\right) = \sum_{k\in\mathbb{Z}^{d-1}} \prod_{1\leq i < j < d} \frac{\alpha_{ij}(X_\phi + X_k)}{\sin\left(\frac{\alpha_{ij}(X_\phi)}{2}\right)}\exp\left(-\frac{1}{4\sigma} \norm{X_\phi + X_k}^2\right)\\
    = &\sum_{k\in\mathbb{Z}^{d-1}} \frac{2b + 2\pi(k_1-k_2)}{\sin\left(b\right)}\prod_{i,j \text{ remaining}} \frac{\alpha_{ij}(X_\phi + X_k)}{\sin\left(\frac{\alpha_{ij}(X_\phi)}{2}\right)}\exp\left(-\frac{1}{4\sigma} \norm{X_\phi + X_k}^2\right),
\end{align}
where $\prod_{i,j \text{ remaining}}$ denotes all of the terms in the product $\prod_{1\leq i < j \leq d}$ except for the $i=1$, $j=2$ term which has now been written explicitly. We now split the sum over $k$ into parts for the cases $k_1 \neq k_2$ and $k_1=k_2$ To streamline the notation the vector $k$ is now parametrised by the two elements $k_1$ and $k_2$ and the remaining $d-3$ dimensional vector $k^\prime$ 
\begin{align}
    \sum_{k^\prime \in\mathbb{Z}^{d-3}} \sum_{k_1=k_2} \frac{2b }{\sin\left(b\right)}\prod_{i,j \text{ remaining}} \frac{\alpha_{ij}(X_\phi + X_k)}{\sin\left(\frac{\alpha_{ij}(X_\phi)}{2}\right)}\exp\left(-\frac{1}{4\sigma} \norm{X_\phi + X_k}^2\right) +\nonumber\\
    \sum_{k^\prime \in\mathbb{Z}^{d-3}} \sum_{k_1\neq k_2} \frac{2b + 2\pi(k_1-k_2)}{\sin\left(b\right)}\prod_{i,j \text{ remaining}} \frac{\alpha_{ij}(X_\phi + X_k)}{\sin\left(\frac{\alpha_{ij}(X_\phi)}{2}\right)}\exp\left(-\frac{1}{4\sigma} \norm{X_\phi + X_k}^2\right).\label{eqn:example-heat-kernel-possion-form-split}
\end{align}
As can be seen, the singularity in the first term now has the expected form $b/\sin(b)$ which converges to $1$ in the limit $b\to 0$ and can easily be extended to a continuous function. For $k_1\neq k_2$ we need to match the term with $k_1=c$, $k_2=d$ with the term with $k_1=d$, $k_2=c$ in order to obtain the cancellation that we need. Let $\hat{k}$ be the vector with elements $c, d, k^\prime$, and $\tilde{k}$ be the vector with elements $d, c, k^\prime$. With this notation, the second term in equation~\eqref{eqn:example-heat-kernel-possion-form-split} becomes 
\begin{align}
    \sum_{k^\prime \in\mathbb{Z}^{d-3}} \sum_{c < d} &\left(\frac{2b + 2\pi(c-d)}{\sin\left(b\right)}\prod_{i,j \text{ remaining}} \frac{\alpha_{ij}(X_\phi + X_{\hat{k}})}{\sin\left(\frac{\alpha_{ij}(X_\phi)}{2}\right)}\exp\left(-\frac{1}{4\sigma} \norm{X_\phi + X_{\hat{k}}}^2\right) \right.+ \nonumber\\
    &\quad\left.\frac{2b + 2\pi(d-c)}{\sin\left(b\right)}\prod_{i,j \text{ remaining}} \frac{\alpha_{ij}(X_\phi + X_{\tilde{k}})}{\sin\left(\frac{\alpha_{ij}(X_\phi)}{2}\right)}\exp\left(-\frac{1}{4\sigma} \norm{X_\phi + X_{\tilde{k}}}^2\right)\right),
\end{align}
and we can now take the limit $b\to 0$ to observe that
\begin{align}
    \lim_{b\to 0} \prod_{i,j \text{ remaining}} \frac{\alpha_{ij}(X_\phi + X_{\hat{k}})}{\sin\left(\frac{\alpha_{ij}(X_\phi)}{2}\right)}\exp\left(-\frac{1}{4\sigma} \norm{X_\phi + X_{\hat{k}}}^2\right)=\lim_{b\to 0} \prod_{i,j \text{ remaining}} \frac{\alpha_{ij}(X_\phi + X_{\tilde{k}})}{\sin\left(\frac{\alpha_{ij}(X_\phi)}{2}\right)}\exp\left(-\frac{1}{4\sigma} \norm{X_\phi + X_{\tilde{k}}}^2\right).
\end{align}
Denoting this limit $L(k)$, we obtain
\begin{align}
  \sum_{k^\prime \in\mathbb{Z}^{d-3}} \sum_{c < d}L(k) \lim_{b\to 0} \left(\frac{2b + 2\pi(c-d)}{\sin\left(b\right)} + \frac{2b + 2\pi(d-c)}{\sin\left(b\right)}\right) = 4\sum_{k^\prime \in\mathbb{Z}^{d-3}} \sum_{c < d}L(k),
\end{align}
where we have cancelled the $(c-d)$ term with the $(d-c)$ term and used the well-known fact that $\lim_{b\to 0} \frac{b}{\sin(b)}$ converges. An identical phenomenon appears when more than $2$ eigenvalues become equal to each other, but proving this directly is considerably more cumbersome.

\bibliography{manuscript}
\end{document}